\documentclass[12pt]{article}

\usepackage[utf8]{inputenc}
\usepackage{fullpage}
\usepackage{amsmath}

\usepackage{amsthm}
\usepackage{mathtools}
\usepackage{bbm}
\usepackage{frcursive}
\usepackage{commath}
\usepackage{mathrsfs}
\usepackage{xcolor}
\usepackage{amssymb}
\usepackage{float}
\usepackage{subcaption}
\usepackage[margin=1in]{geometry} % specify the margin

\usepackage{hyperref}
\usepackage{comment}
\usepackage[UKenglish]{babel}

\usepackage[nodisplayskipstretch]{setspace}
\newcommand{\eps}{\varepsilon}

\newtheorem{theorem}{Theorem}[section]

\newtheorem{lemma}[theorem]{Lemma}

\newtheorem{proposition}[theorem]{Proposition}

\newtheorem{definition}[theorem]{Definition}

\newtheorem{remark}[theorem]{Remark}
\newtheorem{example}[theorem]{Example}
\newtheorem*{example*}{Example}

\begin{document}

\title{Delegate Pricing to Algorithms:\\ When Slow and Steady Wins the Race}

\author{
Galit Askenazi-Golan\textsuperscript{1},
\hspace{0.2cm}
Edward Plumb\textsuperscript{1},\\
Clemens Possnig\textsuperscript{1,2},
\hspace{0.2cm}
Yufei Zhang\textsuperscript{3}
}

\maketitle

\footnotetext[1]{Department of Mathematics, London School of Economics. 
Email: \texttt{g.askenazi-golan, e.plumb @lse.ac.uk}}

\footnotetext[2]{Department of Economics, University of Waterloo. 
Email: \texttt{c.possnig (@) uwaterloo.ca}\\
Financial support from the
Social Sciences and Humanities Research Council of Canada
(SSHRC) through an Insight Development Grant
 is gratefully acknowledged.}

\footnotetext[3]{Department of Mathematics, Imperial College London. 
Email: \texttt{yufei.zhang (@) imperial.ac.uk}}

% \author{Author Name\thanks{}

\begin{abstract}
 We study the strategic design of learning algorithms in a canonical continuous-time pricing game. We introduce a meta-game in which firms select learning rates for a gradient dynamic, evaluating payoffs as the discounted sum of profits accrued along the entire learning path. We uncover a fundamental dichotomy driven by stage-game incentives: if the stage game exhibits strategic substitutability, each firm strictly prefers faster learning. Under strategic complementarity, however, an excessively fast algorithm accelerates the rival's competitive response, destroying transitional profit margins. Even though the competitive price is a dominant action in the stage game and convergence to this price is guaranteed, we prove that firms design sluggish algorithms in equilibrium.
\end{abstract}

\section{Introduction}

The delegation of strategic decision-making to autonomous learning algorithms is fundamentally reshaping oligopolistic competition. Firms increasingly rely on algorithmic agents to set prices and determine output \cite{calvano2020artificial,assad2024algorithmic}. The vast majority of the literature on algorithmic competition satisfies at least one of two facts: (1), the algorithmic learning rule is treated as exogenous, and asymptotic performance is the object of study \cite{banchio2022artificial,possnig2023reinforcement}, or (2), the algorithms studied are learning strategies contingent on opponents' actions \cite{lamba2026pricing}. We depart from both points, by studying the strategic design of learning algorithms constrained to one-shot policies that are guaranteed to converge to a competitive price. Nevertheless, we show that in equilibrium, supra-competitive prices can be sustained during a prolonged transition phase.

This is important from a regulatory perspective: our results imply that to prevent supra-competitive prices from occurring, regulators can neither rely on the fact that algorithms may converge, in the long run, to competitive equilibria; nor on the fact that algorithms may be constrained to learn only simple, myopic rules.

% In practice, learning algorithms are themselves objects of strategic design: firms can deliberately tune hyperparameters, such as learning rates, to control the speed at which their algorithms adapt along a gradient estimate. 

% Furthermore, since learning can be costly while agents adapt, firms have incentives to optimize the profits generated throughout the learning phase, rather than merely its asymptotic outcome.

% This transitional perspective has been largely understudied. We address this gap by considering the strategic design of the learning path itself. Specifically, we ask a fundamental question: 

% \begin{center}
%     \textit{When strategic interaction is delegated to algorithms, what equilibrium behavior emerges in the ex ante stage, when firms choose their algorithms to maximize accumulated payoffs along the learning path?}
%     \end{center}

We study a meta-game of algorithmic adaptation. Firms  compete in a Bertrand oligopoly stage game by delegating their actions to algorithmic agents. These agents update their strategies autonomously in continuous time, following a gradient dynamic that moves in the direction of steepest local payoff increase. The firms themselves strategically choose their algorithm's \textit{learning rate}, a parameter governing the speed of gradient adjustment.\footnote{Gradient-based policy optimisation is a standard component of modern
reinforcement-learning implementations \cite{raffin2021stable}, with
commercial applications including online retail pricing
\cite{liu2019dynamic} and automated bidding in advertising auctions
\cite{mou2022sustainable}, all of which require the choice of a learning rate parameter.} 
The firms maximize accumulated  profits generated throughout the learning phase, rather than merely its asymptotic outcome, as learning can be costly. 

Our central contribution is the discovery that the strategic incentives governing algorithmic design are  dictated by the fundamental nature of the economic interaction within the Bertrand market. By modeling Bertrand competition as a Prisoner's Dilemma, where the baseline stage game guarantees convergence to the competitive Nash equilibrium regardless of learning rates, we isolate the strategic dimension of surplus extraction along the out-of-equilibrium adjustment path. We reveal a stark dichotomy: the meta-game over learning rates can exhibit significantly different equilibrium predictions, depending on the relative gains from undercutting a high-price versus a low-price rival in the stage game.

This comparison of relative gains  pins down whether the stage game exhibits strategic substitutability or complementarity. Strategic substitutability is satisfied if the gain from undercutting a high-price rival is larger than the gain from undercutting a low-price rival. In this scenario, faster learning is an unambiguously optimal strategy, and no equilibrium exists in the meta game. Under strategic complementarity however, Nash equilibria of moderate learning rates can exist.

The mechanism is simple: acceleration towards the low price always has a beneficial direct effect on a firm's flow of profits. However, acceleration also affects the learning trajectory of the competitor's algorithm. When the stage game exhibits strategic substitutes, the acceleration of one firm \textit{slows} the opponent's convergence, allowing for longer-term profitable exploitation of the opponent. As a result, both firms always find it profitable to accelerate, and no equilibrium exists (Theorem \ref{thm:equilibrium}).

Under strategic complementarity, however, the effect of increasing one firm's learning rate  \textit{accelerates} the opponent's learning trajectory. While the direct effect of acceleration is still beneficial for the firm, the reaction of the opponent's trajectory is a disadvantage. In our main theoretical result, we characterise open sets of payoffs and initial conditions under which this opposing effect leads to interior Nash equilibria of the meta-game, where both firms optimally moderate their learning rates (Theorem \ref{thm:interior-NE-informal}).

We further examine how this interior equilibrium learning rate responds to weaker strategic complementarity in the stage game (Proposition~\ref{prop:delta_sensitivity}). The direction of this response depends on the normalized loss a low-pricing firm incurs when its rival switches from high to low pricing. 
When the normalized loss is above a threshold \(\bar t\), weaker complementarity raises the equilibrium learning rate; below it, the equilibrium learning rate falls.

{

The underlying mechanism differs across the two regimes. Above the threshold \(\bar t\), a firm learns slowly primarily to avoid inducing faster adjustment by its rival. Weaker complementarity dampens this strategic feedback, reducing the cost of learning quickly and thereby raising the equilibrium learning rate. Below the threshold, slow learning instead serves primarily to prolong the high-price transition. In this regime, weaker complementarity makes a firm’s high prices less effective at slowing its rival’s convergence. The firm therefore lowers its own learning rate to sustain the gradual transition. Thus, weaker complementarity mitigates the risk of inducing rival acceleration when the relative loss is high, but weakens the ability to sustain slow joint adjustment when the relative loss is low. The equilibrium speed of learning therefore depends on the relative importance of these two strategic motives, rather than on strategic complementarity alone.
}

% This distinction reflects the relative timing of the marginal gains from faster own adjustment and the marginal losses associated with accelerating the rival. Above the threshold, the losses are more back-loaded than the gains; below it, this ordering reverses. Holding the learning rates fixed, weaker complementarity strengthens each firm’s unilateral incentive to accelerate in both regimes. The equilibrium response nevertheless differs: above the threshold, both firms accelerate, bringing the later losses forward and increasing their relative discounted weight; below it, both firms slow down, delaying the later gains and discounting them more heavily. These adjustments restore mutual best responses, showing that the direction of the equilibrium response depends on the timing of marginal payoff effects, not on strategic complementarity alone.

% Algorithmic deceleration acts as a dynamic commitment device. Theoretically, it parallels  strategic delegation and the ``puppy-dog ploy'' of dynamic oligopoly \cite{fudenberg1984fat}, where firms commit to acting ``soft'' to induce a softer competitive response from a rival, reimagined here in the space of algorithmic hyperparameters. 

% Our results provide a novel explanation for supra-competitive pricing in algorithmic markets. Regulators need not exclusively look for complex, collusive reward-punishment schemes in the asymptotic limit. Elevated market prices can emerge simply because firms optimally decelerate learning to delay the onset of the competitive equilibrium.

While we focus our main exposition on oligopolistic pricing, our framework naturally extends to broader classes of economic interaction. In ongoing research, we map qualitatively distinct phenomena in other stage games, demonstrating that choices of learning rates introduce a game of equilibrium selection in coordination games, and generate strict, monotonic incentives to accelerate in zero-sum environments. Furthermore, we establish that near any stable equilibrium, firms always prefer to learn faster, implying that algorithmic sluggishness is fundamentally driven by global transition dynamics rather than local steady-state optimization.

\subsection{Related Literature}

This paper bridges the classic economic theory of strategic delegation with the modern industrial organization of algorithmic pricing, extending literature across economic theory, computer science, and dynamical systems. The closest connections are to work that treats an algorithm, its
initial policy, or its adaptation parameters as a strategic choice. While there is strong emerging interest in the area of strategic choice of algorithms, all existing works we know of \cite{carissimo2025cycles,luo2026algorithmic} take a numerical approach. To the best of our knowledge, we provide the first theoretical framework and analytical results of a meta-game in algorithmic hyperparameters, and consequently also the first characterisation of Nash equilibria in such a setting. 

% \cite{carissimo2025cycles} formulate parameter-choice
% meta-games for Q-learning in congestion games, with payoffs evaluated over the
% repeated learning process. Their numerical analysis finds profitable
% unilateral deviations in the configurations studied. We share their focus on
% incentives at the algorithm-design stage. Our analysis concerns a different source
% of strategic value: delaying a transition that necessarily ends at the competitive
% outcome, rather than exploiting recurrent learning-induced cycles. Within this
% setting, we establish sufficient conditions for a strict symmetric equilibrium over
% the continuum of positive learning rates.

% \cite{luo2026algorithmic} study pricing meta-games in which
% agents choose initial policies together with test-time adaptation rules. They estimate
% empirical payoff matrices and analyse equilibria among candidate meta-strategies,
% including learning-rate choices. According to these estimates,
% smaller learning rates can outperform larger ones. We instead
% fix the initial state and study deterministic projected-gradient dynamics with a
% common competitive terminal outcome for every positive rate profile. Our contribution
% is an analytical characterisation of strategic delay, including equilibria that
% withstand every unilateral change to a positive constant learning rate and a local
% comparative-static reversal governed by stage-game payoffs.

Other theoretical works related to our problem explore the commitment aspect of algorithm design. \cite{lamba2026pricing} analyse a repeated duopoly in which firms choose
algorithms mapping the rival's price into their own.  Their equilibrium analysis shows how these commitments can displace competitive pricing with persistent supracompetitive outcomes. We instead study whether higher prices survive when such long-run outcomes are impossible: every positive learning-rate profile reaches the competitive price in finite time. We show when firms can nevertheless sustain an equilibrium with prolonged periods of elevated prices, while the eventual
competitive outcome remains unchanged.

The interpretation as strategic commitment also connects our results to managerial
delegation \cite{fershtman1987equilibrium} and strategic investment
\cite{fudenberg1984fat}. In algorithmic pricing,
\cite{brown2023competition} endogenise pricing frequencies and show that adopting
slower pricing technology can be profitable. Where their frequency choices alter
the timing of price-setting opportunities, our learning rate scales the magnitude
of continuous gradient adjustment. We identify when this restricted instrument
supports symmetric interior equilibrium choices despite the strict dominance of
the competitive price in the stage game.

Another strand studies supracompetitive behaviour generated by prescribed learning
algorithms \cite{calvano2020artificial,banchio2022artificial}. Our question is not
whether learning sustains a noncompetitive long-run outcome: such an outcome is
ruled out by the chosen stage game and gradient dynamic. Instead, we study how
firms' incentives to select learning speeds affect the transition. The distinction
is between eventual competitive convergence and the discounted profitability of
the path leading to it.

Finally, our framework is related to strategic interaction with adaptive agents
\cite{deng2019strategizing,arslantas2024strategizing}. 
\cite{kolumbus2022auctions} show that users of regret-minimising auction agents can
have incentives to misreport valuations to their own agents. We keep payoff
representations unchanged and vary a learning parameter. Likewise,
opponent-learning-aware methods explicitly modify updates to account for
anticipated changes in rivals' policies \cite{foerster2017learning}. Our agents use
ordinary independent payoff gradients; it is the firms' ex ante choice of a
constant multiplier, rather than an opponent-shaping correction within the update
rule, that strategically influences the rival's trajectory.

\section{The Model}\label{sec:model}

Consider two firms, indexed by $i \in N= \{1, 2\}$, competing repeatedly over prices of a homogeneous good. We model this interaction as a Prisoner's Dilemma played repeatedly in continuous time. Firms choose between maintaining a high, supra-competitive price $p_M$ and undercutting to a low, competitive price $p_C<p_M$. Let $x(t) \in [0,1]$ denote the probability that Firm~$1$ sets the $p_M$ at time $t$, and let $y(t) \in [0,1]$ denote the corresponding probability for Firm~$2$. The state of the market at any time $t$ is fully described by the profile $(x(t), y(t))$ in the state space $S = [0,1]^2$. Let $u_i(x,y)$ represent the expected stage-game profit for firm $i$ given the current pricing state. We refer to the tuple $G=\{N, A, (u_i)_{i\in N}\}$ as the stage game, where $A=\{p_C,p_M\}$ is the action set of each firm.

The normal-form representation of $G$ is:
\begin{equation}
\label{eq:payoff-matrix}
\begin{array}{|l|c|c|}
    \hline
    & {p_M} & {p_C} \\
    \hline
    {p_M} & R,R & S,T \\
    \hline
    {p_C} & T,S & P,P \\
    \hline
\end{array}
\end{equation}
The payoff parameters satisfy $T > R > P > S$, as is standard in the Prisoner's Dilemma: low prices $(p_C, p_C)$ constitute the unique Nash equilibrium in dominant strategies, whilst mutual cooperation $(p_M, p_M)$ yields higher payoffs for both players. 

\subsection{Algorithmic Delegation and Learning Dynamics}

Firms delegate their strategic interactions to autonomous algorithmic agents. These agents update their pricing strategies via a projected gradient dynamic that moves in the direction of steepest local profit increase. 

Crucially, the firms engineer their algorithms by strategically selecting a \textit{learning rate}, $\gamma_i > 0$, which acts as a scalar parameter governing the speed of gradient adjustment. Given a profile of learning rates $\gamma = (\gamma_1, \gamma_2)$, the scaled projected gradient dynamics of the algorithms are:
$$\dot{x}(t) = \Pi_{T_{[0,1]}(x(t))} \left( \gamma_1 \frac{\partial u_1(x(t), y(t))}{\partial x} \right)$$
$$\dot{y}(t) = \Pi_{T_{[0,1]}(y(t))} \left( \gamma_2 \frac{\partial u_2(x(t), y(t))}{\partial y} \right),$$
where $T_{[0,1]}(\cdot)$ denotes the tangent cone, and $\Pi$ is the Euclidean projection onto this closed convex cone, ensuring that the state $(x(t), y(t))$ remains within the bounds of $S$ for all $t \ge 0$.

Because the underlying Bertrand stage game is a dominance-solvable Prisoner's Dilemma, these gradient dynamics guarantee that the market will eventually converge to the unique competitive stage-game Nash equilibrium $(p_C,p_C)$. Note that while the profile of learning rates $(\gamma_1, \gamma_2)$ dictates the exact trajectory and speed of convergence, the stationary points of the dynamics are independent of the learning rates and coincide with the stage game's equilibria.

\subsection{The Meta-Game}

The algorithms will inevitably converge to the competitive equilibrium regardless of the chosen learning rates, but firms accrue payoffs along the entire time the algorithms learn and compete. Hence, the central strategic question for the firms shifts from \textit{what} outcome is reached to the \textit{path} along which it is reached. 

Fix a strictly positive common discount rate $r > 0$ and an initial market condition $(x_0, y_0) \in S$, which may represent a state of historically high prices. For any profile of learning rates $\gamma = (\gamma_1, \gamma_2) \in (0,\infty)^2$, the gradient dynamics yield a unique deterministic learning trajectory. In the ex ante meta-game, each firm $i$ chooses its learning rate $\gamma_i$ from the   set $  (0, \infty)$ to maximize its discounted accumulated profit along this entire learning path:
\begin{align}\label{eq: V_def}
    V_i(\gamma_1, \gamma_2; x_0, y_0) = r \int_0^\infty e^{-rt} u_i(x(t), y(t)) dt
\end{align}

The normalization by $r$  ensures that the meta-game payoff $V_i$ is mapped back into the space of feasible stage-game payoffs. 

The solution concept for this strategic delegation problem is a  Nash equilibrium in the meta-game. 

\begin{definition}\label{def:meta-NE}
A profile $\gamma^* = (\gamma_1^*, \gamma_2^*) \in (0,\infty)^2$ is a \emph{Nash equilibrium of the meta-game} if
\begin{align*}
    V_1(\gamma_1^*, \gamma_2^*;\, x_0, y_0) &\geq V_1(\gamma_1, \gamma_2^*;\, x_0, y_0), \quad \forall \gamma_1 >0;  \\
    V_2(\gamma_1^*, \gamma_2^*;\, x_0, y_0) &\geq V_2(\gamma_1^*, \gamma_2;\, x_0, y_0), \quad \forall \gamma_2>0.
\end{align*}
\end{definition}

When the initial action profile $(x_0,y_0)$ is clear from the context, we omit it for readability.
% \paragraph{Scope of the analysis.}
% In this paper we restrict attention to two-player, two-action stage games. Despite this apparent simplicity, the meta-game exhibits a rich structure that varies qualitatively across game classes: monotone best responses under strategic complementarity (Section Complements), non-monotonic best responses under strategic substitutability (Section Substitutes), payoff discontinuities under multiple equilibria (Section~\ref{sec: BoS}), and arms-race dynamics in zero-sum games (Section Matching Pennies).

\subsection{Example}\label{subsec: PD example}

To illustrate the strategic considerations of the meta game, we consider the following numerical example:
\begin{equation}
\label{eq:pd-example}
\begin{array}{|l|c|c|}
    \hline
    & {p_M} & {p_C} \\
    \hline
    {p_M} & 8,8 & -3,9 \\
    \hline
    {p_C} & 9,-3 & -1,-1 \\
    \hline
\end{array}
\end{equation}

Recall that $x(t)\in [0,1]$ refers to Firm~$1$'s likelihood of choosing $p_M$ in period $t$, and $y(t)\in [0,1]$ that of Firm~$2$. Under the scaled projected gradient dynamics with learning rates $\gamma_1, \gamma_2 > 0$, the evolution of strategies is governed by (see 
\cite{lahkar2008projection}):
\begin{align*}
\dot{x}(t) &= \begin{cases}
-\frac{\gamma_1}{2}(2 - y(t)), & \text{if } x(t) \in (0, 1], \\
0, & \text{if } x(t) = 0,
\end{cases} \quad 
\dot{y}(t) = \begin{cases}
-\frac{\gamma_2}{2}(2 - x(t)), & \text{if } y(t) \in (0, 1], \\
0, & \text{if } y(t) = 0.
\end{cases}
\end{align*}

These dynamics reflect the dominance structure of the game. Since $2 - y(t) > 0$ and $2 - x(t) > 0$ for all feasible states,  
$ 
\dot{x}(t) \leq 0$ and $\dot{y}(t) \leq 0$
for all $(x(t), y(t)) \in [0,1]^2$, with strict inequality in the interior. Thus, both algorithms monotonically decrease the probability assigned to $p_M$, their dominated actions.

Given that each firm has a strictly dominant action $p_C$, intuition might suggest that firms should adopt the highest feasible learning rate to have their algorithms converge rapidly to this action. However, this intuition proves misleading. When Firm~$1$ increases the learning rate, $x(t)$ decreases quicker. Then, $-\frac{\gamma_2}{2}(2 - x(t))$ is smaller sooner, meaning that Firm~$2$ decreases $y(t)$ faster as a result of Firm~$1$'s acceleration. The intuition for this is due to the fact that the payoff gradient from switching to a low price $p_C$ is \textit{steeper} when facing a low-price rival firm than a high-pricing one.

 This gap is reflected in Firm~$2$'s gradient dynamics, causing them to switch faster to $p_C$ when Firm~$1$ switches faster to $p_C$.  Firm~$2$ accelerating the transition to $p_C$ reduces Firm~$1$'s payoff. 

These observations are reflected in our numerical analysis:

\begin{figure}[ht]
            \centering
        \includegraphics[width=0.5\linewidth]{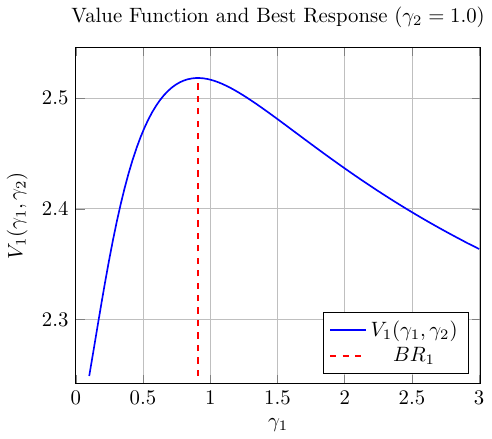}
        \caption{\footnotesize Firm~$1$'s value $V_1(\gamma_1,\gamma_2)$ as function of $\gamma_1$, given $\gamma_2=1$ and $(x_0,y_0)=(0.9,0.9)$.}
    \label{fig: V1}
\end{figure}

\begin{example}[Non-monotonicity of Optimal Learning Rates]\label{ex:nonmonotone}
Consider the initial condition $(x_0, y_0) = (0.9, 0.9)$ and discount rate $r = 1$. Fix Firm~$2$'s learning rate at $\gamma_2 = 1$. As shown in Figure~\ref{fig: V1}, numerical analysis reveals  that Firm~$1$'s discounted utility $V_1(\gamma_1, 1)$ is maximised in the interior at $\gamma_1^* \approx 0.9082$.
\end{example}

\section{Pricing, Fast and Slow}\label{sec: PD}

 Going back to the general stage game $G$, with the payoff matrix \eqref{eq:payoff-matrix}.  We introduce the following parametrisation to facilitate the analysis:
\begin{align*}
\kappa &\coloneqq P - S > 0, \\
\delta &\coloneqq (T - R) - (P - S).
\end{align*}
Here $\kappa$ is the absolute gain from pricing low against a low-price rival, and $\delta$ measures the degree of strategic complementarity discussed informally in the introduction. Specifically, $\delta$ compares the volume gained from undercutting a high-price rival $(T-R)$ with the loss avoided by matching a low-price rival $(P-S)$. 

The sign of $\delta$ partitions the strategic landscape. When $\delta \ge 0$ (strategic \emph{substitutes}), the volume gains from undercutting a cooperator are large; here, an acceleration by Firm~1 towards $p_C$ \emph{reduces} the instantaneous value of an undercut for Firm~2, slowing the opponent's convergence. This acts as an added strategic benefit of speed. 

When $\delta < 0$ (strategic \emph{complements}), failing to match a low price is highly punitive ($P-S > T-R$), a feature characteristic of markets with capacity constraints, tipping points, or strong brand loyalty. In this regime, the reverse holds: Firm~$1$'s acceleration strictly increases the incentive for Firm~$2$ to undercut, which accelerates the rival's descent and is strictly costly to Firm~$1$. These two cases are treated in turn.

Given initial conditions $(x_0, y_0) \in (0,1)^2$ and learning rates $\gamma_1, \gamma_2 > 0$, the projected gradient dynamics in our general case are (see 
\cite{lahkar2008projection}):
\begin{align*}
\dot{x}(t) &= \begin{cases}
-\frac{\gamma_1}{2}(\kappa + \delta y(t)), & \text{if } x(t) \in (0,1], \\
0, & \text{if } x(t) = 0,
\end{cases} \\
\dot{y}(t) &= \begin{cases}
-\frac{\gamma_2}{2}(\kappa + \delta x(t)), & \text{if } y(t) \in (0,1], \\
0, & \text{if } y(t) = 0.
\end{cases}
\end{align*}

Since $\kappa > 0$ and $\kappa + \delta = T - R > 0$, both $\dot{x}(t)$ and $\dot{y}(t)$ are strictly negative in the interior of the state space. Define the boundary hitting time:
\begin{align*}
H \coloneqq \sup\{t \geq 0 : x(t) \in (0,1) \text{ and } y(t) \in (0,1)\},
\end{align*}
which represents the first time at which the trajectory reaches the boundary of $[0,1]^2$. 
%The monotonicity of the dynamics ensures that $(x(t), y(t)) \to (0,0)$ as $t \to H$.

\subsection{Strategic Substitutes}

When $\delta \ge 0$, the strategic and direct effects of speed reinforce one another. The following result makes the mechanism precise: a faster Firm~$1$ reaches the competitive price sooner while \emph{slowing} the opponent.

\begin{proposition}[Trajectory Sensitivity: Substitutes]\label{prop:sensitivity}
Suppose $\delta \geq 0$ and $(x_0, y_0) \in (0,1]^2$. Then for all $t \in (0,H)$:
\begin{align*}
\frac{\partial x^\gamma(t)}{\partial \gamma_1} < 0, \quad \frac{\partial y^\gamma(t)}{\partial \gamma_1} \ge 0, \\
\frac{\partial x^\gamma(t)}{\partial \gamma_2} \ge 0, \quad \frac{\partial y^\gamma(t)}{\partial \gamma_2} < 0,
\end{align*}
where $\gamma = (\gamma_1, \gamma_2)$ denotes the profile of learning rates, and strict inequalities are recovered under $\delta>0$.
\end{proposition}

When Firm~$1$ increases $\gamma_1$ and converges faster to $p_C$ ($\partial x/\partial\gamma_1 < 0$), the opponent's convergence slows ($\partial y/\partial\gamma_1 \ge 0$). Both effects raise Firm~$1$'s payoff, and the benefit persists along the entire path. Hence, Firm~$1$ can only gain from increasing their learning rate, as further outlined below.

\begin{proposition}
    [Monotone Best Responses]\label{prop:monotone}
Suppose $\delta \geq 0$ and $x_0, y_0 \in (0,1]$. Then for any $\bar{\gamma}_2 \in (0,\infty)$, the function $\gamma_1 \mapsto V_1(\gamma_1, \bar{\gamma}_2)$ is strictly increasing on $(0, \infty)$.
\end{proposition}

The monotonicity pins down the meta-game outcome. If there is a maximum learning rate, then it is optimal for both players to play this. If there is none, then the best response correspondence is empty. 

Hence, markets where undercutting a cooperator is more tempting than undercutting a low-price firm validate a basic intuition:  dominant strategies in the stage game lead to a race to the bottom in the meta game of adjustment speeds.

\begin{theorem}
    [Equilibrium Characterisation]\label{thm:equilibrium}
Under strategic substitutability ($\delta \geq 0$), for all initial conditions $(x_0, y_0) \in (0, 1]^2$, the meta-game does not admit a Nash equilibrium in  $(0,\infty)^2$.
\end{theorem}

\subsection{Strategic Complements}

We now examine the case where $\delta < 0$, corresponding to strategic complementarity in the stage game. Firstly, the impact of one's own speed on the rival's learning trajectory is opposite to that of the previous section.

\begin{proposition}[Trajectory Sensitivity: Complements]\label{prop:sensitivity_delta_negative}
Suppose $\delta < 0$ and $(x_0, y_0) \in (0,1]^2$. Then for all $t \in (0,H)$:
\begin{align*}
\frac{\partial x^\gamma(t)}{\partial \gamma_1} < 0, \quad \frac{\partial y^\gamma(t)}{\partial \gamma_1} \le 0, \\
\frac{\partial x^\gamma(t)}{\partial \gamma_2} \le 0, \quad \frac{\partial y^\gamma(t)}{\partial \gamma_2} < 0,
\end{align*}
where $\gamma = (\gamma_1, \gamma_2)$ denotes the profile of learning rates. 
%Tand strict inequalities are recovered under $\delta<0$.
\end{proposition}

When the stage game exhibits strategic complementarity, an increase in Firm~$1$'s rate accelerates Firm~$2$'s transition away from high prices. If Firm~$1$ learns too rapidly, they reach $p_C$; once there, Firm~$1$ ceases to adapt while Firm~$2$'s strategy continues to evolve. By arriving at the boundary too soon, Firm~$1$ minimises the time spent exploiting Firm~$2$'s pre-convergence cooperation. Moderating the learning rate thus prolongs the opponent's exploitability, and the best response becomes non-monotone, as already seen in Example~\ref{ex:nonmonotone}.

The central question is whether this non-monotonicity sustains an interior equilibrium. Our main result answers this affirmatively.

\begin{theorem}[Interior equilibrium under complementarity]\label{thm:interior-NE-informal}
Suppose $\delta < 0$. Define the normalized temptation spread $q=\frac{T-R}{P-S}$. There is a strictly increasing, positive function $F(q)$, and a singularity $\bar t>0$, such that when
\begin{align*}
    \frac{T-P}{P-S}\in (F(q),\infty)\setminus \{\bar t\},
\end{align*}
a symmetric equilibrium in finite learning rates exists. Specifically, there is an open set  $\mathcal N \subset (0,1)$ satisfying that for every initial condition $x_0 = y_0 \in \mathcal N$ there exists $\eta^*(x_0)$ such that for every discount rate $r > 0$, the meta-game admits a finite, strictly positive symmetric profile $(\gamma^*(r), \gamma^*(r))$ that is a strict global symmetric Nash equilibrium in learning rates, where $\gamma^*(r) = \eta^*(x_0)r$.
\end{theorem}

 We prove the result in two stages, developed in full in Appendices~\ref{sec:undiscounted} and~\ref{sec:discounted}. 

 First, we exploit the fact that the discounted payoff has a homogeneity property ensuring that payoffs on learning rate profiles $(\gamma_1,\gamma_2)$ only through their ratio $\gamma_1/\gamma_2$, and on the discount factor only through the effective rate $r/\gamma_2$. Second, we solve the simpler problem of the undiscounted limit game ($r=0$), and transport the result back to our discounted target payoffs via a uniform convergence argument.

The equilibrium rests on a trade-off between the local, myopic incentive to optimize against the rival's current price and the global, strategic cost of accelerating the rival's convergence to the competitive limit price. We show that this is summarized by the comparison between what we call the `normalized temptation spread' $\frac{T-R}{P-S}$, and the `acceleration loss' $\frac{T-P}{P-S}$. The former represents the relative gains from undercutting a high-price rival. The latter measures the relative loss of undercutting a high price rival who switches to a low price. One can think of this quantity as measuring the relative loss when faced with an accelerating rival. When this opportunity cost of destroying the transitional phase of high prices ($\frac{T-P}{P-S}$) outweighs the short-term benefit of rapid optimization along the gradient ($\frac{T-R}{P-S}$), an equilibrium exists; firms optimally design sluggish algorithms.

The proportionality $\gamma^{*}(r) \propto r$ is a direct consequence of the homogeneity property of the discounted payoff : the equilibrium learning rate precisely tracks the discount rate. Patient decision makers optimally learn more slowly.

Figure~\ref{fig: br intersect} illustrates the theorem for our example developed in Subsection~\ref{subsec: PD example} with $r = 1$: the best response correspondences intersect on the diagonal at an interior symmetric profile $\gamma^* \approx 0.479$.

\begin{figure}
    \centering
    \includegraphics[width=0.6\linewidth]{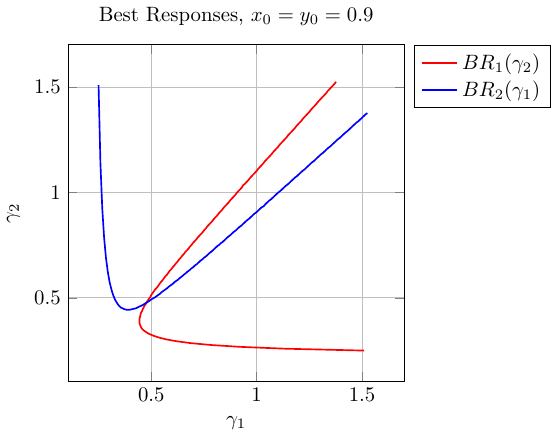}
    \caption{\footnotesize Best response correspondences for both players with symmetric initial condition $x_0=y_0=0.9$. The intersection at $\gamma_1=\gamma_2\approx 0.479$ provides numerical evidence for a symmetric Nash equilibrium in learning rates.}
    \label{fig: br intersect}
\end{figure}

Finally, we study how the symmetric equilibrium in learning rates responds to changes in the complementarity parameter $\delta<0$. As in Theorem~\ref{thm:interior-NE-informal}, the acceleration loss plays a crucial role.

\begin{proposition}
\label{prop:delta_sensitivity}

Suppose the conditions of Theorem \ref{thm:interior-NE-informal} hold. 
% Consider a marginal increase in $\delta$ (a weakening of strategic complementarity), holding $\kappa$ and $T-P$ constant. 
There is 
  an open set $\mathcal N\subset(0,1)$ and 
a threshold $\bar t>0$ such that
for every initial
condition $x_0=y_0\in \mathcal N$ and every discount rate $r>0$,
the symmetric
equilibrium learning rate $\gamma^*(r,x_0;\delta)$   is locally
continuously differentiable in $\delta$ (holding $\kappa$, $T-P$,
$x_0$, and $r$ fixed), and    its response to a marginal increase in $\delta$
(a weakening of strategic complementarity) satisfies
\begin{equation}
 \frac{\partial \gamma^*(r,x_0;\delta)}{\partial \delta}    \begin{cases}
        >0, & \frac{T-P}{P-S}> \bar t,
        \\
        <0, & \frac{T-P}{P-S}<\bar t.
    \end{cases}
\end{equation}
\end{proposition}

 The threshold $\bar t$ partitions the meta-game into two distinct strategic regimes:

\paragraph{High Relative Loss ($\frac{T-P}{P-S}> \bar t$).}
When the relative loss from a rival's acceleration is severe, the firm engineers a sluggish algorithm primarily to avoid steepening the rival's learning gradient. In this regime, as strategic complementarity weakens ($\delta \to 0$), the rival's updating speed becomes less sensitive to the firm's price cuts. Because the firm's own acceleration now triggers a weaker retaliatory acceleration from the rival, the strategic risk of learning fast is mitigated. The firm therefore responds by increasing its optimal learning rate.

\paragraph{Low Relative Loss ($\frac{T-P}{P-S}< \bar t$).}
Conversely, when the penalty from an accelerating rival is relatively small, the firm's deceleration is driven instead by the broader desire to jointly prolong the cooperative transition phase. The comparative statics reverse. As $\delta \to 0$, the firm's high prices act as a less effective drag on the rival's convergence. To compensate for this weakened transmission mechanism and ensure the market descent remains gradual, the firm is forced to maintain elevated prices for longer, adopting an \textit{even more sluggish} algorithm.

\section{Conclusion}\label{sec: conclusion}

Our central finding challenges the conventional intuition that firms unambiguously prefer faster algorithms to rapidly optimize towards a dominant action. In markets exhibiting strategic complementarity, an excessively fast algorithm rapidly undercuts the rival but inadvertently steepens the rival's payoff gradient, accelerating their descent toward the competitive equilibrium and destroying transitional profit margins. Consequently, firms optimally design sluggish algorithms. The parametric condition for algorithmic sluggishness arises naturally in vulnerable markets characterized by stark competitive tipping points, loyal consumer bases, or capacity constraints.

This deliberate deceleration uncovers a distinct mechanism for anti-competitive harm. Even when standard learning algorithms are guaranteed to eventually converge to a competitive Nash equilibrium, firms can extract substantial profits above the competitive benchmark simply by tuning hyperparameters to delay this convergence. For competition authorities, this implies that a purely steady-state analysis of algorithmic pricing is insufficient. Algorithms that appear benign in the limit can be engineered to act as sophisticated commitment devices that artificially prolong states of historically high prices. Detecting anti-competitive behavior in algorithmic markets may therefore require scrutinizing the \emph{ex ante} strategic choices that govern the speed of market adjustment.

The meta-game framework developed here naturally extends beyond pricing games, mapping a rich taxonomy of algorithmic behavior. In ongoing research, we show that in zero-sum settings, the strictly monotone incentive to accelerate traps players in an algorithmic arms race. In coordination games, asymmetric learning rates can serve as an equilibrium selection device. 

A natural avenue for future research is to endogenize the initialization of the algorithms. In an enlarged meta-game where firms jointly choose both their initial states and their learning rates, the incentive to learn slowly might establish initialization near mutual cooperation as a robust equilibrium. Ultimately, as strategic decisions are increasingly delegated to autonomous agents, integrating the strategic design of learning rules into economic theory provides a rich and necessary agenda for future work. 

\begingroup
\singlespacing
\normalsize
\bibliographystyle{plain}
\bibliography{reference}
\endgroup
% \bibliographystyle{plain}
% \bibliography{reference}

\appendix

\section{Appendix}\label{sec: Appendix}

\subsection{Preliminary Results of Section~\ref{sec: PD}}

\subsubsection{Proof of Proposition~\ref{prop:sensitivity}}

\begin{proof}
We only prove the case with $\delta>0$. The results for  $\delta=0$ follow  from similar arguments, since the dynamics of $x$ and $y$ are decoupled.

By the theorem on parameter dependence of differential equations \cite[Chapter V, Theorem 3.1]{hartman2002ode}, the partial derivatives $u(t) \coloneqq \frac{\partial x^\gamma(t)}{\partial \gamma_1}$ and $v(t) \coloneqq \frac{\partial y^\gamma(t)}{\partial \gamma_1}$ satisfy the variational equations:
\begin{align}
\dot{u}(t) &= -\frac{1}{2}(\kappa + \delta y(t)) - \frac{\gamma_1 \delta}{2} v(t), \quad u(0) = 0, \label{eq:variation_u}\\
\dot{v}(t) &= -\frac{\gamma_2 \delta}{2} u(t), \quad v(0) = 0. \label{eq:variation_v}
\end{align}

\textit{Step 1: Local behaviour.} We first establish that $u < 0$ and $v > 0$ on an interval $(0, \varepsilon)$ for some $\varepsilon > 0$. Since $\kappa > 0$, $\delta > 0$, and $y(0) \in (0,1)$, continuity ensures that for sufficiently small $\varepsilon > 0$:
\begin{align*}
-\frac{1}{2}(\kappa + \delta y(t)) - \frac{\gamma_1 \delta}{2} v(t) < 0 \quad \text{for all } t \in (0, \varepsilon].
\end{align*}
From equation \eqref{eq:variation_u} with $u(0) = 0$, this implies $u(t) < 0$ for $t \in (0, \varepsilon]$. Substituting into \eqref{eq:variation_v} yields $\dot{v}(t) > 0$, hence $v(t) > 0$ for $t \in (0, \varepsilon]$.

\textit{Step 2: Global behaviour.} We prove $u(t) < 0$ for all $t \in (0,H)$. Suppose, for contradiction, that there exists $t_0 > 0$ such that $u(t_0) \geq 0$. Let:
\begin{align*}
\bar{t} = \inf\{t > 0 : u(t) \geq 0\}.
\end{align*}
By Step 1, $\bar t>0$. 
By continuity, $u(\bar{t}) = 0$ and $u(t) < 0$ for $t \in (0, \bar{t})$. From \eqref{eq:variation_v} with $\delta > 0$, we have $v(t) \geq 0$ on $[0, \bar{t}]$. But then \eqref{eq:variation_u} implies:
\begin{align*}
\dot{u}(\bar{t}) = -\frac{1}{2}(\kappa + \delta y(\bar{t})) - \frac{\gamma_1 \delta}{2} v(\bar{t}) < 0,
\end{align*}
contradicting the fact that $u$ reaches zero from below at $\bar{t}$. Therefore $u(t) < 0$ for all $t \in (0,H)$, which by \eqref{eq:variation_v} implies $v(t) > 0$.

\textit{Step 3: Symmetric argument.} By the symmetry of the game, interchanging the roles of players yields
$ 
\frac{\partial x^\gamma(t)}{\partial \gamma_2} > 0 $ and $\text{and} \quad \frac{\partial y^\gamma(t)}{\partial \gamma_2} < 0 $ 
for all $t \in (0,H)$.
\end{proof}

\subsubsection{Proof of Proposition~\ref{prop:monotone}}

\begin{proof}
Define $J_1=V_1/r$.
Fix $x_0, y_0 \in (0,1]$ and $\gamma_2 > 0$. By the Leibniz integral rule, 
\begin{align*}
\frac{\partial J_1}{\partial \gamma_1}(\gamma_1, \gamma_2) = \int_0^H e^{-rt} \frac{\partial}{\partial \gamma_1}\big[(S-P)x^\gamma(t) + (T-P)y^\gamma(t) - \delta x^\gamma(t)y^\gamma(t)\big]\, dt + R,
\end{align*}
where the remainder term $R$ depends on which player hits zero first:
$$
R\coloneqq \begin{cases}
    \int_{T_y}^{T_{y,x}} e^{-rt}(S{-}P)\frac{\partial x(t)}{\partial \gamma_1}\, dt, &  T_y \le  T_x, 
    \\
    \int_{T_x}^{T_{x,y}} e^{-rt}(T{-}P)\frac{\partial y(t)}{\partial \gamma_1}\, dt, &  T_x < T_y.
\end{cases}
$$
It suffices to show that both terms are non-negative, with at least one strictly positive.

On $[0,H]$ we sign the integrand by Proposition~\ref{prop:sensitivity}. Since $S < P$,
\begin{align*}
(S - P)\frac{\partial x^\gamma(t)}{\partial \gamma_1} > 0,
\end{align*}
and since $\delta \geq 0$,
\begin{align*}
-\delta\, \frac{\partial x^\gamma(t)}{\partial \gamma_1}\, y^\gamma(t) \geq 0.
\end{align*}
The terms involving $\partial_{\gamma_1} y^\gamma$ combine as
\begin{align*}
\big[(T - P) - \delta x^\gamma(t)\big] \frac{\partial y^\gamma(t)}{\partial \gamma_1}
\geq \big[T - P - \delta\big] \frac{\partial y^\gamma(t)}{\partial \gamma_1}
= (R - S) \frac{\partial y^\gamma(t)}{\partial \gamma_1} \geq 0,
\end{align*}
using $x^\gamma(t) \leq 1$, the identity $\delta = S + T - P - R$, $R > S$, and $\partial_{\gamma_1} y^\gamma \geq 0$; the inequality is strict when $\delta > 0$, while at $\delta = 0$ we have $\partial_{\gamma_1} y^\gamma \equiv 0$ and the term vanishes. In either case the $(S-P)\partial_{\gamma_1}x^\gamma$ contribution is strictly positive, so the integrand is strictly positive and
\begin{align}\label{eq:before_hitting_align}
\int_0^{H} e^{-rt} \frac{\partial}{\partial \gamma_1}\big[(S - P)x^\gamma(t) + (T - P)y^\gamma(t) - \delta x^\gamma(t) y^\gamma(t)\big]\, dt > 0.
\end{align}

We now consider the remainder term. If $T_x = T_y$, then $R = 0$ and the claim follows from~\eqref{eq:before_hitting_align}. Otherwise we sign $R$ using the post-boundary sensitivities. Suppose first $H = T_y < T_x$. We have $x^\gamma(t) = x(\gamma, T_y) - \frac{\gamma_1 \kappa}{2}(t - T_y)$ on $[T_y, T_{y,x}]$. Differentiating in $\gamma_1$ with $t$ fixed and applying the chain rule to $T_y(\gamma)$,
\begin{align*}
\frac{\partial x^\gamma(t)}{\partial \gamma_1} = \frac{\partial x}{\partial \gamma_1}(\gamma, T_y) + \Big[\frac{\partial x}{\partial t}(\gamma, T_y) + \frac{\gamma_1 \kappa}{2}\Big]\frac{\partial T_y}{\partial \gamma_1} - \frac{\kappa}{2}(t - T_y).
\end{align*}
At the hitting time the interior velocity is $\frac{\partial x}{\partial t}(\gamma, T_y) = -\frac{\gamma_1}{2}(\kappa + \delta y(T_y)) = -\frac{\gamma_1 \kappa}{2}$, since $y(T_y) = 0$, so the bracket vanishes and
\begin{align*}
\frac{\partial x^\gamma(t)}{\partial \gamma_1} = \frac{\partial x}{\partial \gamma_1}(\gamma, T_y) - \frac{\kappa}{2}(t - T_y) < 0,
\end{align*}
by Proposition~\ref{prop:sensitivity} together with $\kappa > 0$ and $t \geq T_y$. With $S - P < 0$ this gives $R > 0$. Suppose instead $H = T_x < T_y$. Analogously to the case for $x^{\gamma}(t)$, we have $y^\gamma(t) = y(\gamma, T_x) - \frac{\gamma_2 \kappa}{2}(t - T_x)$ on $[T_x, T_{x,y}]$, and the boundary velocity $-\frac{\gamma_2 \kappa}{2}$ is independent of $\gamma_1$. The same chain-rule cancellation, now using $\frac{\partial y}{\partial t}(\gamma, T_x) = -\frac{\gamma_2}{2}(\kappa + \delta x(T_x)) = -\frac{\gamma_2 \kappa}{2}$, leaves
\begin{align*}
\frac{\partial y^\gamma(t)}{\partial \gamma_1} = \frac{\partial y}{\partial \gamma_1}(\gamma, T_x) \geq 0,
\end{align*}
by Proposition~\ref{prop:sensitivity}. With $T > P$ this gives $R \geq 0$. In every case $\partial_{\gamma_1} J_1 > 0$, which completes the proof.
\end{proof}

\subsubsection{Proof of Theorem \ref{thm:equilibrium}}
\begin{proof} Proposition~\ref{prop:monotone} establishes that for any $\gamma_2 > 0$, Firm~$1$'s payoff function $J_1(\gamma_1, \gamma_2)$ is strictly increasing in $\gamma_1$ on the interval $(0, \infty)$. 
% \textit{Part (i): Bounded rates.} When $\gamma_i \in (0, \bar{\gamma}]$ for both players, by the strict monotonicity from Proposition~\ref{prop:monotone}, Firm~$1$'s best response to any $\gamma_2 \in (0, \bar{\gamma}]$ is $\mathrm{BR}_1(\gamma_2) = \bar{\gamma}.$ By symmetry of the game structure and initial conditions, Firm~$2$'s best response satisfies $\mathrm{BR}_2(\gamma_1) = \bar{\gamma}.$
% Therefore, $(\bar{\gamma}, \bar{\gamma})$ is the unique fixed point of the best response correspondence. Uniqueness follows from the strict monotonicity: any profile $(\gamma_1, \gamma_2)$ with $\gamma_i < \bar{\gamma}$ for some $i \in \{1, 2\}$ cannot be an equilibrium, as player $i$ has a profitable deviation to $\bar{\gamma}$.
% \textit{Part (ii): Unbounded rates.} When learning rates are unbounded, 
Consider any candidate equilibrium $(\gamma_1, \gamma_2) \in (0, \infty)^2$. Without loss of generality, $\gamma_1 \leq \gamma_2$. By Proposition~\ref{prop:monotone}, for some $\varepsilon > 0$, Firm~$1$'s payoff satisfies $J_1(\gamma_1 + \varepsilon, \gamma_2) > J_1(\gamma_1, \gamma_2)$.
Hence $\gamma_1$ is not a best response to $\gamma_2$. Since this holds for arbitrary $(\gamma_1, \gamma_2)$, no Nash equilibrium exists. The best response correspondence is empty: $\mathrm{BR}_i(\gamma_{-i}) = \varnothing$ for all $\gamma_{-i} \in (0, \infty)$ and $i \in \{1, 2\}$.
\end{proof}

\subsubsection{Proof of Proposition~\ref{prop:sensitivity_delta_negative}}
\begin{proof}
The proof is analogous to that of Proposition~\ref{prop:sensitivity}, with the sign change in sensitivities only due to the sign change of $\delta$.
\end{proof}

\subsection{Undiscounted Game}\label{sec:undiscounted}

We first note the following scaling results:

\begin{lemma}[Trajectory rescaling]\label{lem:traj-rescaling}
Fix a game $G$, learning rates $\gamma=(\gamma_1,\gamma_2)$, an initial condition $(x_0,y_0)\in S$, and $\lambda>0$. Writing $(x(\cdot;\gamma),y(\cdot;\gamma))$ for the trajectory under the scaled projected gradient dynamics with rates $\gamma$ from $(x_0,y_0)$, the trajectory under rates $\lambda\gamma$ from the same initial condition is its time-rescaling:
\begin{align*}
x(\tau;\lambda\gamma)=x(\lambda\tau;\gamma),\qquad y(\tau;\lambda\gamma)=y(\lambda\tau;\gamma),\qquad \tau\ge 0.
\end{align*}
\end{lemma}

\begin{proof}
Define $\tilde x(\tau):=x(\lambda\tau;\gamma)$ and $\tilde y(\tau):=y(\lambda\tau;\gamma)$. By the chain rule,
\begin{align*}
\frac{d\tilde x}{d\tau}(\tau)=\lambda\,\dot x(\lambda\tau;\gamma)=\lambda\,\Pi_{T_{\Delta(A_1)}(\tilde x(\tau))}\bigl(\gamma_1\,U_1(\tilde x(\tau),\tilde y(\tau))\bigr)=\Pi_{T_{\Delta(A_1)}(\tilde x(\tau))}\bigl(\lambda\gamma_1\,U_1(\tilde x(\tau),\tilde y(\tau))\bigr),
\end{align*}
where the last equality uses the positive homogeneity of the projector onto a closed convex cone \cite[Proposition 28.22]{bauschke2011convex}, valid because $T_{\Delta(A_1)}(z)$ is a non-empty closed convex cone for every $z\in\Delta(A_1)$. The analogous identity holds for $\tilde y$, and $(\tilde x(0),\tilde y(0))=(x_0,y_0)$. By uniqueness of Carathéodory solutions, $(\tilde x,\tilde y)$ is the trajectory under rates $\lambda\gamma$ from $(x_0,y_0)$.
\end{proof}

\begin{proposition}[Discounted payoff scaling]\label{prop:scaling}
For every $\gamma_1,\gamma_2,r,\lambda>0$ and every initial condition $(x_0,y_0)\in S$,
\begin{align*}
V_1(\lambda\gamma_1,\lambda\gamma_2;\lambda r;x_0,y_0)=V_1(\gamma_1,\gamma_2;r;x_0,y_0).
\end{align*}
\end{proposition}

\begin{proof}
By Lemma~\ref{lem:traj-rescaling}, $x(\tau;\lambda\gamma)=x(\lambda\tau;\gamma)$ and $y(\tau;\lambda\gamma)=y(\lambda\tau;\gamma)$. Substituting into the definition of $V_1$ and changing variables to $t=\lambda\tau$,
\begin{align*}
V_1(\lambda\gamma_1,\lambda\gamma_2;\lambda r) &=\lambda r\int_0^\infty e^{-\lambda r\tau}\,u_1\bigl(x(\lambda\tau;\gamma),y(\lambda\tau;\gamma)\bigr)\,d\tau \\
&=r\int_0^\infty e^{-rt}\,u_1\bigl(x(t;\gamma),y(t;\gamma)\bigr)\,dt \\
&=V_1(\gamma_1,\gamma_2;r).
\end{align*}
\end{proof}

Fix a symmetric initial condition $x_0 = y_0 \in (0,1]$. Define the undiscounted excess payoff
\begin{align*}
J^0_1(\gamma_1, \gamma_2; x_0) \;:=\; \int_0^{H'} \bigl[u_1(x(t), y(t)) - P\bigr]\, dt,
\end{align*}
where $H'$ is the time at which both coordinates have been absorbed at zero. Since the integrand vanishes for $t \geq H'$, this integral is finite, and by dominated convergence
\begin{align*}
\lim_{r \to 0} \frac{V_1(\gamma_1, \gamma_2; r; x_0, x_0) - P}{r} \;=\; J^0_1(\gamma_1, \gamma_2; x_0).
\end{align*}

The undiscounted payoff satisfies a homogeneity property in the learning rates analogous to the discounted scaling identity.

\begin{lemma}\label{lem:undiscounted-scaling}
For every $\gamma_1,\gamma_2,\lambda>0$ and every initial condition $(x_0,y_0)\in S$,
\begin{align*}
J^0_1(\lambda\gamma_1,\lambda\gamma_2;x_0,y_0)=\lambda^{-1}\,J^0_1(\gamma_1,\gamma_2;x_0,y_0).
\end{align*}
In particular, setting $\lambda=1/\gamma_2$ and $\rho=\gamma_1/\gamma_2$,
\begin{align}\label{eq:scaling-rho-1}
J^0_1(\gamma_1,\gamma_2;x_0,y_0)=\gamma_2^{-1}\,J^0_1(\rho,1;x_0,y_0).
\end{align}
\end{lemma}

\begin{proof}
The relation $J_1^0(\gamma_1,\gamma_2)=\lim_{r\to 0^+}\bigl(V_1(\gamma_1,\gamma_2;r)-P\bigr)/r$ holds at every fixed initial condition, since $u_1-P$ is bounded and vanishes beyond the finite absorption time $H'$. Applying Proposition~\ref{prop:scaling} with discount rate $r/\lambda$ gives $V_1(\lambda\gamma_1,\lambda\gamma_2;r)=V_1(\gamma_1,\gamma_2;r/\lambda)$, whence, substituting $s=r/\lambda$,
\begin{align*}
J_1^0(\lambda\gamma_1,\lambda\gamma_2)=\lim_{r\to 0^+}\frac{V_1(\gamma_1,\gamma_2;r/\lambda)-P}{r}=\lim_{s\to 0^+}\frac{V_1(\gamma_1,\gamma_2;s)-P}{\lambda s}=\lambda^{-1}\,J_1^0(\gamma_1,\gamma_2).
\end{align*}
The identity~\eqref{eq:scaling-rho-1} follows by taking $\lambda=1/\gamma_2$.
\end{proof}

The factor $\gamma_2^{-1}$ in~\eqref{eq:scaling-rho-1} does not affect first-order conditions in $\gamma_1$ at the symmetric profile $\gamma_1 = \gamma_2$, so we focus on the function $\rho \mapsto J^0_1(\rho, 1; x_0)$, whose critical point at $\rho = 1$ characterises a symmetric Nash equilibrium of the undiscounted meta-game.

\begin{theorem}\label{thm:undiscounted-FOC}
Suppose $\delta < 0$ and $x_0 = y_0 \in (0,1]$. Then $\rho = 1$ is a critical point of $\rho \mapsto J^0_1(\rho, 1; x_0)$ if and only if
\begin{align}\label{eq:FOC-undiscounted}
(T-P)\left[\frac{x_0^2}{2\kappa} - \frac{x_0}{\delta} + \frac{\kappa}{\delta^2}\ln\!\left(1 + \frac{\delta x_0}{\kappa}\right)\right] + x_0^2 \;=\; 0.
\end{align}
\end{theorem}

\begin{proof}[Proof of Theorem~\ref{thm:undiscounted-FOC}]
Throughout, fix $x_0 = y_0 \in (0, 1]$ and $\delta < 0$. Set $\gamma_2 = 1$ so $\gamma_1 = \rho$, and write $K_0 := \kappa + \delta x_0 > 0$.

Let $x(t;\rho),y(t;\rho)$ denote the projected trajectory.
Write $H_1(\rho)$ for the first time at which at least one
coordinate reaches zero, and $H_2(\rho)$ for the first time
at which both coordinates have reached zero. Let
$\widehat x(t;\rho),\widehat y(t;\rho)$ solve the
unprojected system
\begin{align*}
\dot{\widehat x}
=-\frac{\rho}{2}(\kappa+\delta\widehat y),
\qquad
\dot{\widehat y}
=-\frac12(\kappa+\delta\widehat x),
\qquad
\widehat x(0;\rho)=\widehat y(0;\rho)=x_0.
\end{align*}
The affine vector field gives joint smoothness in $(t,\rho)$.

At $\rho=1$, symmetry gives
$\widehat x(t;1)=\widehat y(t;1)=:z(t)$. Writing
$K(t):=\kappa+\delta z(t)$, we obtain
\begin{align*}
K(t)=K_0e^{-\delta t/2},
\qquad
T^*=\frac{2}{\delta}\ln\!\left(\frac{K_0}{\kappa}\right)>0,
\qquad
z(T^*)=0.
\end{align*}
Thus $H_1(1)=H_2(1)=T^*$.

Put $h:=\rho-1$. Since $\dot z=-K/2<0$, smooth dependence
keeps both unprojected time derivatives bounded above by
a negative constant on a fixed interval $[0,T^*+\eta]$,
for some $\eta>0$ and all $\rho$ sufficiently close to $1$.
Their positive initial values and negative values at
$T^*+\eta$ persist, while their values at $T^*$ are
$O(|h|)$. The mean value theorem therefore places both
first zeros at $T^*+O(|h|)$. Until the earlier zero,
both coordinates lie in $(0,x_0]$, so the projected and
unprojected paths coincide. Consequently,
$H_1(\rho)=T^*+O(|h|)$, and smoothness near $(T^*,1)$
also gives
\begin{align*}
x(H_1(\rho);\rho)+y(H_1(\rho);\rho)=O(|h|).
\end{align*}

Extend the excess stage payoff to $\mathbb R^2$ by its
polynomial expression
\begin{align*}
\Delta u_1(x,y):=-\kappa x+(T-P)y-\delta xy,
\qquad
F(t,\rho):=
\Delta u_1(\widehat x(t;\rho),\widehat y(t;\rho)).
\end{align*}
Split the actual payoff as
\begin{align}\label{eq:J0-decomp}
J_1^0(\rho,1)
&=
\underbrace{
\int_0^{H_1(\rho)}F(t,\rho)\,dt
}_{=:\,\mathcal I(\rho)}
+
\underbrace{
\int_{H_1(\rho)}^{H_2(\rho)}
\Delta u_1(x(t;\rho),y(t;\rho))\,dt
}_{=:\,\mathcal B(\rho)}.
\end{align}

At the first absorption, the remaining coordinate is
$O(|h|)$ and subsequently decreases at rate $\kappa/2$
or $\rho\kappa/2$, both bounded away from zero near
$\rho=1$. Thus $H_2(\rho)-H_1(\rho)=O(|h|)$.
Throughout this boundary phase the state remains
$O(|h|)$ from the origin, so the excess payoff is
$O(|h|)$ as well. Multiplying this bound by the duration
gives $\mathcal B(\rho)=O(h^2)$.
At $\rho=1$ the boundary phase has zero length, so
$\mathcal B(1)=0$. Therefore, for some constant $C_B>0$,
\begin{align*}
\left|
\frac{\mathcal B(1+h)-\mathcal B(1)}{h}
\right|
\leq C_B|h|\longrightarrow0
\qquad\text{as }h\to0,\ h\neq0.
\end{align*}
This proves $\mathcal B'(1)=0$ from both sides without
differentiating either absorption time.

Moreover, $F$ is smooth and $F(T^*,1)=0$. Hence
$|F(t,\rho)|\leq C_F(|t-T^*|+|h|)$ near $(T^*,1)$
for some constant $C_F>0$, and
\begin{align*}
\left|
\mathcal I(\rho)-\int_0^{T^*}F(t,\rho)\,dt
\right|
&\leq
C_F|H_1(\rho)-T^*|
\bigl(|H_1(\rho)-T^*|+|h|\bigr)\\
&=O(h^2).
\end{align*}
The difference vanishes at $\rho=1$; dividing by $h$
shows that $\mathcal I$ has the same derivative there
as the smooth fixed-horizon integral.

Define the unprojected sensitivities
\begin{align*}
u(t):=
\left.\partial_\rho\widehat x(t;\rho)\right|_{\rho=1},
\qquad
v(t):=
\left.\partial_\rho\widehat y(t;\rho)\right|_{\rho=1}.
\end{align*}
Differentiating the unprojected equations gives
\begin{align*}
\dot u=-\frac12K-\frac{\delta}{2}v,
\qquad
\dot v=-\frac{\delta}{2}u,
\qquad
u(0)=v(0)=0.
\end{align*}
Writing $p:=u+v$ and $w:=u-v$, addition and subtraction
yield
\begin{align*}
\dot p+\frac{\delta}{2}p=-\frac12K,
\qquad
\dot w-\frac{\delta}{2}w=-\frac12K.
\end{align*}
Solving with $p(0)=w(0)=0$ gives
\begin{align*}
p(t)=-\frac{K_0t}{2}e^{-\delta t/2},
\qquad
w(t)=-\frac{K_0}{\delta}
\sinh\!\left(\frac{\delta t}{2}\right).
\end{align*}
Since
$\partial_\rho F(t,1)=-K(t)u(t)+(T-S-K(t))v(t)$,
differentiation under the fixed-horizon integral yields
\begin{align}
\mathcal I'(1)
&=\int_0^{T^*}
\bigl[-K(t)u(t)+(T-S-K(t))v(t)\bigr]\,dt,
\label{eq:interior-derivative}\\
\frac{\partial J_1^0}{\partial\rho}(1,1)
&=\mathcal I'(1)+\mathcal B'(1)=\mathcal I'(1).
\label{eq:derivative-target}
\end{align}

Substituting $u = (p + w)/2$ and $v = (p - w)/2$, the integrand becomes $\frac{T - S}{2}(p - w) - K p$. We need three integrals:
\begin{align*}
I_p \;:=\; \int_0^{T^*}\!p(t)\, dt, \qquad I_w \;:=\; \int_0^{T^*}\!w(t)\, dt, \qquad I_{Kp} \;:=\; \int_0^{T^*}\!K(t)\, p(t)\, dt.
\end{align*}

Firstly, using $p(t) = -(K_0 t/2)\, e^{-\delta t/2}$,
\begin{align*}
I_p \;&=\; \int_0^{T^*}\!\!-\frac{K_0\, t}{2}\, e^{-\delta t/2}\, dt =\; -\frac{K_0}{2}\int_0^{T^*}\!t\, e^{-\delta t/2}\, dt.
\end{align*}
Integration by parts with $u = t$, $dv = e^{-\delta t/2}\, dt$ (so $du = dt$ and $v = -(2/\delta) e^{-\delta t/2}$) gives
\begin{align*}
\int_0^{T^*}\!t\, e^{-\delta t/2}\, dt \;&=\; \left[-\frac{2 t}{\delta}\, e^{-\delta t/2}\right]_0^{T^*} + \frac{2}{\delta}\int_0^{T^*}\!e^{-\delta t/2}\, dt \\[4pt]
&=\; -\frac{2 T^*}{\delta}\, e^{-\delta T^*/2} + \frac{2}{\delta}\!\left[-\frac{2}{\delta}\, e^{-\delta t/2}\right]_0^{T^*} \\[4pt]
&=\; -\frac{2 T^*}{\delta}\, e^{-\delta T^*/2} + \frac{4}{\delta^2}\bigl(1 - e^{-\delta T^*/2}\bigr).
\end{align*}
Substituting $e^{-\delta T^*/2} = \kappa/K_0$ and $T^* = (2/\delta)\ln(K_0/\kappa)$,
\begin{align*}
\int_0^{T^*}\!t\, e^{-\delta t/2}\, dt \;&=\; -\frac{2}{\delta}\cdot\frac{2}{\delta}\ln\!\left(\frac{K_0}{\kappa}\right)\cdot\frac{\kappa}{K_0} + \frac{4}{\delta^2}\!\left(1 - \frac{\kappa}{K_0}\right) \\[4pt]
&=\; -\frac{4\kappa}{\delta^2 K_0}\ln\!\left(\frac{K_0}{\kappa}\right) + \frac{4 x_0}{\delta K_0},
\end{align*}
using $K_0 - \kappa = \delta x_0$ in the last step. Multiplying by $-K_0/2$,
\begin{align*}
I_p \;&=\; -\frac{K_0}{2}\!\left[-\frac{4\kappa}{\delta^2 K_0}\ln\!\left(\frac{K_0}{\kappa}\right) + \frac{4 x_0}{\delta K_0}\right] \\[4pt]
&=\; \frac{2\kappa}{\delta^2}\ln\!\left(\frac{K_0}{\kappa}\right) - \frac{2 x_0}{\delta}.
\end{align*}

Secondly, using $w(t) = -(K_0/\delta) \sinh(\delta t/2)$,
\begin{align*}
I_w \;&=\; \int_0^{T^*}\!\!-\frac{K_0}{\delta}\sinh\!\left(\frac{\delta t}{2}\right) dt \\[4pt]
&=\; -\frac{2 K_0}{\delta^2}\!\left[\cosh\!\left(\frac{\delta T^*}{2}\right) - 1\right].
\end{align*}
Using $\cosh(\delta T^*/2) = (e^{\delta T^*/2} + e^{-\delta T^*/2})/2 = (K_0/\kappa + \kappa/K_0)/2$,
\begin{align*}
\cosh\!\left(\frac{\delta T^*}{2}\right) - 1 \;&=\; \frac{1}{2}\!\left(\frac{K_0}{\kappa} + \frac{\kappa}{K_0}\right) - 1 =\; \frac{\delta^2 x_0^2}{2 K_0 \kappa},
\end{align*}
again using $K_0 - \kappa = \delta x_0$. Hence,
\begin{align*}
I_w \;&=\; -\frac{2 K_0}{\delta^2}\cdot\frac{\delta^2 x_0^2}{2 K_0 \kappa} =\; -\frac{x_0^2}{\kappa}.
\end{align*}

Finally, using $K(t) p(t) = K_0\, e^{-\delta t/2} \cdot \bigl(-K_0 t/2\bigr)\, e^{-\delta t/2} = -(K_0^2 t/2)\, e^{-\delta t}$,
\begin{align*}
I_{Kp} \;&=\; \int_0^{T^*}\!\!-\frac{K_0^2 t}{2}\, e^{-\delta t}\, dt =\; -\frac{K_0^2}{2}\int_0^{T^*}\!t\, e^{-\delta t}\, dt.
\end{align*}
Integration by parts with $u = t$, $dv = e^{-\delta t}\, dt$ (so $du = dt$ and $v = -(1/\delta) e^{-\delta t}$) gives
\begin{align*}
\int_0^{T^*}\!t\, e^{-\delta t}\, dt \;&=\; \left[-\frac{t}{\delta}\, e^{-\delta t}\right]_0^{T^*} + \frac{1}{\delta}\int_0^{T^*}\!e^{-\delta t}\, dt \\[4pt]
&=\; -\frac{T^*}{\delta}\, e^{-\delta T^*} + \frac{1}{\delta}\!\left[-\frac{1}{\delta}\, e^{-\delta t}\right]_0^{T^*} \\
&=\; -\frac{T^*}{\delta}\, e^{-\delta T^*} + \frac{1}{\delta^2}\bigl(1 - e^{-\delta T^*}\bigr).
\end{align*}
Substituting $e^{-\delta T^*} = (e^{-\delta T^*/2})^2 = \kappa^2/K_0^2$ and $T^* = (2/\delta)\ln(K_0/\kappa)$,
\begin{align*}
\int_0^{T^*}\!t\, e^{-\delta t}\, dt \;&=\; -\frac{1}{\delta}\cdot\frac{2}{\delta}\ln\!\left(\frac{K_0}{\kappa}\right)\cdot\frac{\kappa^2}{K_0^2} + \frac{1}{\delta^2}\!\left(1 - \frac{\kappa^2}{K_0^2}\right) \\[4pt]
&=\; -\frac{2\kappa^2}{\delta^2 K_0^2}\ln\!\left(\frac{K_0}{\kappa}\right) + \frac{K_0^2 - \kappa^2}{\delta^2 K_0^2}.
\end{align*}
Multiplying by $-K_0^2/2$,
\begin{align*}
I_{Kp} \;&=\; -\frac{K_0^2}{2}\!\left[-\frac{2\kappa^2}{\delta^2 K_0^2}\ln\!\left(\frac{K_0}{\kappa}\right) + \frac{K_0^2 - \kappa^2}{\delta^2 K_0^2}\right] \\[4pt]
&=\; \frac{\kappa^2}{\delta^2}\ln\!\left(\frac{K_0}{\kappa}\right) - \frac{K_0^2 - \kappa^2}{2\delta^2}.
\end{align*}
Combining,
\begin{align*}
\mathcal{I}'(1) \;=\; \frac{T - S}{2}(I_p - I_w) - I_{Kp},
\end{align*}
with
\begin{align*}
I_p - I_w \;=\; -\frac{2 x_0}{\delta} + \frac{2\kappa}{\delta^2}\ln\!\left(\frac{K_0}{\kappa}\right) + \frac{x_0^2}{\kappa}.
\end{align*}
The logarithmic coefficient is $(T - S)\kappa/\delta^2 - \kappa^2/\delta^2 = (T - P)\kappa/\delta^2$, using $T - S - \kappa = T - P$. Expanding $K_0^2 - \kappa^2 = 2\delta\kappa x_0 + \delta^2 x_0^2$ and $T - S = (T - P) + \kappa$, the non-logarithmic terms simplify to $-(T - P) x_0/\delta + (T - P) x_0^2/(2\kappa) + x_0^2$. Hence
\begin{align*}
\frac{\partial J^0_1}{\partial \rho}(1, 1) \;=\; \mathcal{I}'(1) + \mathcal{B}'(1) \;=\; (T - P)\!\left[\frac{x_0^2}{2 \kappa} - \frac{x_0}{\delta} + \frac{\kappa}{\delta^2}\ln\!\left(1 + \frac{\delta x_0}{\kappa}\right)\right] + x_0^2.
\end{align*}
Setting this to zero gives~\eqref{eq:FOC-undiscounted}.
\end{proof}

\begin{theorem}\label{thm:undiscounted-global-max}
    Suppose $\delta < 0$ and $x_0 = y_0 \in (0,1]$. Then $\rho \mapsto J^0_1(\rho, 1; x_0)$ admits at most one critical point on $(0,\infty)$. If such a critical point exists, then it is the unique strict global maximiser. In particular, if Equation~\eqref{eq:FOC-undiscounted} holds, then $\rho = 1$ is the unique strict global maximiser and  $\partial_{\rho\rho}J_1^0(1,1;x_0)<0$.
\end{theorem}

\begin{proof}
    At some point in this proof, we will pass to considering the problem spatially. We first make sure that this is well defined. Note that the trajectory $(x(t),y(t))$ satisfies $\dot x/\dot y = \rho(\kappa+\delta y)/(\kappa+\delta x)$ on the interior. Setting $\psi(z) = \kappa z + \tfrac{\delta}{2}z^2$, so that $\psi'(z) = \kappa + \delta z$, we have
\begin{align*}
\frac{d}{dt}\bigl[\psi(x(t)) - \rho\,\psi(y(t))\bigr] = (\kappa+\delta x)\dot x - \rho(\kappa+\delta y)\dot y = 0,
\end{align*}
so $\psi(x(t)) - \rho\,\psi(y(t))$ is conserved along trajectories. Evaluating at $t=0$ with $x(0)=y(0)=x_0$ fixes the constant, giving $\psi(x(t))-\rho\,\psi(y(t)) = (1- \rho)\psi(x_0)$ in the interior. 

Since $\delta<0$ and $\kappa+\delta>0$,
$\psi'(z)=\kappa+\delta z\geq\kappa+\delta>0$ on $[0,x_0]$.
Let $\psi^{-1}\colon[0,\psi(x_0)]\to[0,x_0]$ denote the
inverse of this restriction. For $\rho>0$ and $y\in[0,x_0]$,
\begin{align*}
\rho\psi(y)+(1-\rho)\psi(x_0)
=
\psi(x_0)-\rho\bigl[\psi(x_0)-\psi(y)\bigr]
\leq\psi(x_0).
\end{align*}
Thus the expression
\begin{align*}
x(y,\rho)
:=
\psi^{-1}\!\left(
\max\bigl\{0,\rho\psi(y)+(1-\rho)\psi(x_0)\bigr\}
\right)
\end{align*}
is well defined.

This is the projected spatial trajectory up to the first
time $y$ reaches zero. The value $x(0,\rho)$ is
taken at the first absorption of $y$; any subsequent
segment along $y=0$ contributes nothing to the integral
of $y(t)$ used below. 

We can now proceed by considering:
\begin{align*}
    J_1^0(\rho) = \int_0^{H'} \big[ u_1(x(t), y(t)) - P \big] \, dt = \int_0^{H'} \big[ -x(t)(\kappa + \delta y(t)) + (T-P)y(t) \big] \, dt,
\end{align*}
where $\kappa = P-S > 0$, and $H'$ is the time when both coordinates are absorbed at zero. 

Recall that, while $x(t) > 0$, the dynamics for $x$ are $\dot{x}(t) = -\frac{\rho}{2}(\kappa + \delta y(t))$ regardless of whether $y(t) =0 $ or not. Therefore, $-x(t)(\kappa + \delta y(t)) = \frac{2}{\rho} x(t)\dot{x}(t) = \frac{1}{\rho}\frac{d}{dt}(x^2(t))$. Furthermore, if $x(t) = 0$, then we still have $-x(t)(\kappa + \delta y(t)) = \frac{1}{\rho}\frac{d}{dt}(x^2(t))$, and so this equality holds for all $t \in [0, H']$. Hence, 
\begin{align*}
    \int_0^{H'} -x(t)(\kappa + \delta y(t)) \, dt &= \frac{1}{\rho}\int_0^{H'} \frac{d}{dt}(x^2(t))  \, dt =\frac{1}{\rho} \big[ x(H')^2 - x(0)^2 \big] = -\frac{x_0^2}{\rho}.
\end{align*}
Therefore, 
\begin{align*}
    J_1^0(\rho)  = -\frac{x_0^2}{\rho} + (T-P)\int_0^{H'} y(t) dt.
\end{align*}

For the term $\int_0^{H'} y(t) \, dt$, notice that once $y(t)$ hits $0$ at $T_y \le H'$, the integrand vanishes. Thus, we only integrate up to $T_y$. During $t \in [0, T_y)$, $y(t) > 0$ and $\dot{y}(t) = -\frac{1}{2}(\kappa + \delta x(t))$. Changing variables to $y$ with $dt = \frac{-2}{\kappa + \delta x} dy$ we integrate from $x_0$ down to $0$ and so change the sign of the differential:
\begin{align*}
    \int_0^{H'} y(t) \, dt = 2 \int_0^{x_0} \frac{y}{\kappa + \delta x(y, \rho)} \, dy.
\end{align*}

Define
\begin{align*}
y_H(\rho):=
\begin{cases}
0, & 0<\rho\leq1,\\[2pt]
\psi^{-1}\!\left((1-\rho^{-1})\psi(x_0)\right),
& \rho>1.
\end{cases}
\end{align*}
For $\rho>1$, this is the unique value of $y$ at which
$x$ is first absorbed. Splitting the integral where the
projection binds gives:
\begin{align*}
    J_1^0(\rho) = -\frac{x_0^2}{\rho} + 2(T-P) \left[ \int_0^{y_H(\rho)} \frac{y}{\kappa} \, dy + \int_{y_H(\rho)}^{x_0} \frac{y}{\kappa + \delta x(y, \rho)} \, dy \right].
\end{align*}
We differentiate with respect to $\rho$. By Leibniz's rule, the boundary evaluations at $y_H(\rho)$ cancel since $\kappa + \delta x(y_H(\rho), \rho) = \kappa + 0 = \kappa$. Therefore,
\begin{align*}
    \frac{d}{d\rho} J_1^0(\rho) = \frac{x_0^2}{\rho^2} + 2(T-P) \int_{y_H(\rho)}^{x_0} \frac{-\delta y}{(\kappa + \delta x)^2} \frac{\partial x}{\partial \rho} \, dy.
\end{align*}
Differentiating $\psi(x(y,\rho)) = \rho\,\psi(y) + (1-\rho)\psi(x_0)$ with respect to $\rho$ gives $\psi'(x) \frac{\partial x}{\partial \rho} = \psi(y) - \psi(x_0)$. Since $\psi'(z) = \kappa + \delta z$, this implies $\frac{\partial x}{\partial \rho} = \frac{\psi(y) - \psi(x_0)}{\kappa + \delta x}$. Substituting this yields:
\begin{align*}
    \frac{d}{d\rho} J_1^0(\rho) = \frac{x_0^2}{\rho^2} + 2(T-P) \int_{y_H(\rho)}^{x_0} \frac{-\delta y (\psi(y) - \psi(x_0))}{(\kappa + \delta x(y,\rho))^3} \, dy.
\end{align*}

To prove global optimality, we define the scaled derivative $G(\rho) = \rho^2 \frac{d}{d\rho} J_1^0(\rho)$ and show it is strictly decreasing. In the integral of $G(\rho)$, we perform a continuous change of variables back to $x = x(y, \rho)$. Holding $\rho$ fixed, the invariant relation gives $\psi(y) = \frac{1}{\rho}\psi(x) + (1-\frac{1}{\rho})\psi(x_0)$, yielding $\psi'(y)dy = \frac{1}{\rho}\psi'(x)dx$, which means $dy = \frac{\kappa + \delta x}{\rho(\kappa + \delta y)}dx$. Furthermore, $\psi(y) - \psi(x_0) = \frac{1}{\rho}(\psi(x) - \psi(x_0))$.

The integration limits $y \in [y_H(\rho), x_0]$ map to $x \in [x_{end}(\rho), x_0]$, where $x_{end}(\rho) = 0$ if $\rho \ge 1$, and $\psi(x_{end}) = (1-\rho)\psi(x_0)$ if $\rho \le 1$. Applying this transformation, the scaling factors of $\rho$ cancel out:
\begin{align*} 
G(\rho) &= x_0^2 + 2(T-P) \int_{x_{end}(\rho)}^{x_0} \frac{-\delta y}{(\kappa + \delta x)^3} \left(\frac{\psi(x) - \psi(x_0)}{\rho}\right) \rho^2 \left(\frac{\kappa + \delta x}{\rho(\kappa + \delta y)}\right) dx \nonumber \\
&= x_0^2 + 2(T-P) \int_{x_{end}(\rho)}^{x_0} \underbrace{\left( \frac{-\delta y(x, \rho)}{\kappa + \delta y(x, \rho)} \right)}_{:= A(y(x, \rho))} \underbrace{\left( \frac{\psi(x) - \psi(x_0)}{(\kappa + \delta x)^2} \right)}_{:= B(x)} \, dx.
\end{align*}
Applying Leibniz's rule to $G$,
\begin{align*}
G'(\rho)=2(T-P)\left[-x_{end}'(\rho)\,A(y(x_{end}(\rho),\rho)) B(x_{end}(\rho))+\int_{x_{end}(\rho)}^{x_0}A'(y(x,\rho))\,\partial_\rho y(x,\rho)\,B(x)\,dx\right].
\end{align*}

For $\rho \geq 1$, Firm~$1$ reaches $x = 0$ no later than Firm~$2$ reaches $y = 0$, so the interior phase covers all of $x \in (0, x_0]$. Hence $x_{end}(\rho) \equiv 0$ on $[1, \infty)$, and in particular $x_{end}'(\rho) = 0$ on $(1, \infty)$. For $\rho \in (0, 1)$, Firm~$2$ reaches $y = 0$ before $x$ does, and the interior phase terminates at $x_{end}(\rho) > 0$, the value of $x$ at the instant $y$ vanishes. By construction $y(x_{end}(\rho), \rho) = 0$, so $A(y(x_{end}(\rho), \rho)) = 0$. 

Therefore, in either case, we have $-x_{end}'(\rho)\,A(y(x_{end}(\rho),\rho)) B(x_{end}(\rho)) =0$, and so:
\begin{align*}
G'(\rho)&=2(T-P)\int_{x_{end}(\rho)}^{x_0}A'(y(x,\rho))\,\partial_\rho y(x,\rho)\,B(x)\,dx \\
&= 2(T-P)\int_{x_{end}(\rho)}^{x_0} \frac{-\delta\kappa}{(\kappa+\delta y(x,\rho))^2}\cdot\frac{\psi(x_0)-\psi(x)}{\rho^2(\kappa+\delta y(x,\rho))}\cdot\frac{\psi(x)-\psi(x_0)}{(\kappa+\delta x)^2}\,dx \\
&= \frac{2(T-P)\delta \kappa}{\rho^2} \int_{x_{end}(\rho)}^{x_0} \frac{\bigl(\psi(x) - \psi(x_0)\bigr)^2}{(\kappa + \delta x)^2(\kappa + \delta y(x, \rho))^3}\,dx,
\end{align*}
Every factor in this expression is positive except $\delta$: the constant $2(T-P)/\rho^2 > 0$ since $T > P$; the numerator $(\psi(x)-\psi(x_0))^2 \ge 0$ as a square, and is strictly positive for $x < x_0$ since $\psi$ is strictly increasing; the denominator $(\kappa+\delta x)^2(\kappa+\delta y(x,\rho))^3 > 0$ because $\kappa+\delta z \ge \kappa+\delta = T-R > 0$ for all $z \in [0,1]$; and $\kappa > 0$. Hence the integral has the sign of $\delta < 0$, and is strictly negative whenever $x_{end}(\rho) < x_0$, giving $G'(\rho) < 0$. Therefore, as $x_0 = y_0 > 0$, $G(\rho)$ is strictly monotonically decreasing in $\rho \in (0, \infty)$. 

Since $G$ is strictly decreasing on $(0, \infty)$, it has at most one zero, and hence $\rho \mapsto J_1^0(\rho, 1; x_0)$ has at most one critical point on $(0, \infty)$. When such a critical point $\rho^*$ exists, $G(\rho) > 0$ on $(0, \rho^*)$ and $G(\rho) < 0$ on $(\rho^*, \infty)$, so $J_1^0$ is strictly increasing on $(0, \rho^*)$ and strictly decreasing on $(\rho^*, \infty)$; thus $\rho^*$ is the unique strict global maximiser. In particular, if Equation~\eqref{eq:FOC-undiscounted} holds, then $\rho=1$ is a critical point of $\rho \mapsto J^0_1(\rho, 1; x_0)$, and so it is the unique strict global maximiser of $J_1^0$ on $(0,\infty)$.

Finally, we record the second-order information at a critical point. Differentiating $G(\rho)=\rho^2\,\partial_\rho J_1^0(\rho,1;x_0)$ gives $G'(\rho)=2\rho\,\partial_\rho J_1^0(\rho,1;x_0)+\rho^2\,\partial_{\rho\rho}J_1^0(\rho,1;x_0)$. At a critical point $\rho^*$ we have $\partial_\rho J_1^0(\rho^*,1;x_0)=0$, so $\partial_{\rho\rho}J_1^0(\rho^*,1;x_0)=G'(\rho^*)/(\rho^*)^2$. Since $x_{\mathrm{end}}(\rho)<x_0$ for every $\rho\in(0,\infty)$, the strict inequality $G'(\rho)<0$ holds throughout $(0,\infty)$, and therefore $\partial_{\rho\rho}J_1^0(\rho^*,1;x_0)<0$. When Equation~\eqref{eq:FOC-undiscounted} holds, $\rho^*=1$, giving $\partial_{\rho\rho}J_1^0(1,1;x_0)<0$.
\end{proof}

We now record the facts about the first-order map needed in the discounted analysis. Define
\begin{align}\label{eq:E-def}
E(x_0) := (T-P)\left[\frac{x_0^2}{2\kappa} - \frac{x_0}{\delta} + \frac{\kappa}{\delta^2}\ln\!\left(1 + \frac{\delta x_0}{\kappa}\right)\right] + x_0^2,
\end{align}
so that, by Theorem~\ref{thm:undiscounted-FOC}, $E(x_0) = \partial_\rho J_1^0(1,1;x_0)$, and $\rho=1$ is a critical point of the undiscounted game precisely when $E(x_0)=0$. It is convenient to factor out a positive constant. Writing $\kappa = P-S$ and recalling $\delta<0$, $T>P$, a direct computation gives
\begin{align}\label{eq:E-f-relation}
E(x_0) = \frac{\kappa(T-P)}{\delta^2}\, f(x_0), \qquad \frac{\kappa(T-P)}{\delta^2} > 0,
\end{align}
where
\begin{align}\label{eq:f-def}
f(x) := \frac{\delta x}{2(P-T)(P-S)}\left(2(T-P) - \delta x\,\frac{T+P-2S}{P-S}\right) - \ln\!\left(\frac{P-S}{\delta x + P-S}\right).
\end{align}
Thus $E$ and $f$ share the same sign and the same zeros on $(0,1]$, and the following lemma may be proved for $f$. We emphasise that the \emph{uniqueness} of the critical point of $J_1^0$ in $\rho$ for a fixed initial condition is already supplied by Theorem~\ref{thm:undiscounted-global-max} (strict monotonicity of $G(\rho)$); the lemma below instead controls the dependence of the first-order quantity $E$ on the initial condition $x_0$.

\begin{lemma}[Root and transversality of $E$ in $x_0$]\label{lem:E-root}
Fix $\delta<0$. The map $E$ in \eqref{eq:E-def} has at most one zero in $(0,1]$, and is strictly negative to the right of any such zero. In particular, if $\bar x\in(0,1)$ is a zero of $E$, then $E(x_0)<0$ for all $x_0\in(\bar x,1]$.
\end{lemma}

\begin{proof}
By \eqref{eq:E-f-relation} it suffices to prove the statement for $f$ defined in \eqref{eq:f-def}. Note that $f(0) = 0$, and one verifies that
\begin{align*}
f'(x) &= -\frac{\delta}{P-S} - \frac{\delta^2 x(T+P-2S)}{(P-T)(P-S)^2} + \frac{\delta}{\delta x + P-S},\\[4pt]
f''(x) &= \frac{\delta^2(T+P-2S)}{(T-P)(P-S)^2} - \frac{\delta^2}{(\delta x + P - S)^2}, &
f'''(x) &= \frac{2\delta^3}{(\delta x + P - S)^3} < 0,
\end{align*}
where the sign of $f'''$ follows from $\delta < 0$ and $\delta x + P-S > 0$ for all $x \in [0,1]$. Therefore $f'(0) = 0$, $f''(0) = \frac{2\delta^2}{(T-P)(P-S)} > 0$, and $f''$ is strictly decreasing on $[0,1]$. In particular, $f(x), f'(x), f''(x) > 0$ for all sufficiently small $x > 0$.

Now suppose that there exists $\bar x \in (0,1]$ such that $f(\bar x) = 0$. Since $f(0) = 0$ and $f(\bar x) = 0$, Rolle's theorem yields $x_1 \in (0,\bar x)$ with $f'(x_1) = 0$. Applying Rolle's theorem again to $f'$, using $f'(0) = 0$ and $f'(x_1) = 0$, yields $x_2 \in (0, x_1)$ with $f''(x_2) = 0$. Since $f''$ is strictly decreasing, this zero is unique, and $f''(x) < 0$ for all $x > x_2$. Consequently $f'$ is strictly decreasing on $(x_2,1]$. Since $f'(x_1) = 0$ and $x_1 > x_2$, it follows that $f'(x) < 0$ for all $x \in (x_1,1]$, so $f$ is strictly decreasing on $(x_1,1]$. Since $f(\bar x) = 0$ and $\bar x \in (x_1,1]$, we conclude that $f(x) < 0$ for all $x \in (\bar x,1]$. This argument also shows the zero is unique: a second zero in $(0,1]$ would, by the same reasoning, force $f$ strictly negative immediately to the right of the smaller zero yet equal to zero at the larger, a contradiction.
\end{proof}

\begin{remark}\label{rem:E-root-exists}
Since $f(0)=f'(0)=0$ and $f''(0) = \tfrac{2\delta^2}{(T-P)(P-S)}>0$, we have $E(x_0)>0$ for all sufficiently small $x_0>0$. Hence, by the intermediate value theorem, $E$ has a (necessarily unique, by Lemma~\ref{lem:E-root}) root $\bar x\in(0,1)$ whenever $E(1)<0$. 
\end{remark}

Before proceeding to the next section, we will note the following result:

\begin{lemma}\label{lem:J0-x0-uniform}
Fix $\delta<0$. There is a constant $C_0>0$, depending only on the game parameters, such that for all $x_0\in(0,1]$ and all $\rho\ge 1$,
\begin{align*}
\Bigl|J_1^0(\rho,1;x_0)-\tfrac{T-P}{\kappa}x_0^2\Bigr|\le \frac{C_0}{\rho}.
\end{align*}
Consequently, for every $\bar x\in(0,1)$,
\begin{align*}
\sup_{\rho\ge 1}\bigl|J_1^0(\rho,1;x_0)-J_1^0(\rho,1;\bar x)\bigr|\xrightarrow{x_0\to\bar x}0.
\end{align*}
\end{lemma}

\begin{proof}
From the spatial representation in the proof of Theorem~\ref{thm:undiscounted-global-max}, for $\rho\ge 1$ Firm~$1$'s coordinate reaches zero first, at the value $y=y_H(\rho)$ determined by $\psi(y_H)=(1-1/\rho)\psi(x_0)$, and
\begin{align*}
J_1^0(\rho,1;x_0)=-\frac{x_0^2}{\rho}+(T-P)\left[\frac{y_H^2}{\kappa}+2\int_{y_H}^{x_0}\frac{y}{\kappa+\delta x(y,\rho)}\,dy\right].
\end{align*}
As $y$ is strictly decreasing on the interior and $x_0 > 0$, we have that $y_H \leq x_0$. We then apply the mean value theorem to $\psi$ on $[y_H,x_0]$ and using the bound $\psi'(z)=\kappa+\delta z\ge\kappa+\delta=T-R>0$ on $[0,1]$:
\begin{align*}
\psi(x_0)-\psi(y_H)=\psi'(\xi)\,(x_0-y_H)\ge (T-R)\,(x_0-y_H)\qquad\text{for some }\xi\in(y_H,x_0).
\end{align*}
so that, after substituting the defining relation $\psi(x_0)-\psi(y_H)=\tfrac{1}{\rho}\psi(x_0)$ and using $\psi(x_0)\le \psi(1)$ by monotonicity,
\begin{align*}
0\le x_0-y_H\le\frac{\psi(x_0)-\psi(y_H)}{T-R}=\frac{\psi(x_0)}{\rho\,(T-R)}\le\frac{\psi(1)}{(T-R)\,\rho}.
\end{align*}
Hence $|y_H^2-x_0^2|\le 2(x_0-y_H)\le \frac{2\psi(1)}{(T-R)\rho}$, and, using $\kappa+\delta x\ge T-R$ on the range of integration, $0\le 2\int_{y_H}^{x_0}\frac{y}{\kappa+\delta x}\,dy\le \frac{2}{T-R}\int_{y_H}^{x_0}y\,dy\le\frac{2(x_0-y_H)}{T-R}\le\frac{2\psi(1)}{(T-R)^2\rho}$. Letting $\rho\to\infty$, $y_H\to x_0$ and the integral vanishes, so $\lim_{\rho\to\infty}J_1^0(\rho,1;x_0)=\tfrac{T-P}{\kappa}x_0^2$. Combining the two $O(1/\rho)$ estimates with $x_0^2/\rho\le 1/\rho$ gives the stated bound with $C_0=1+2(T-P)\psi(1)\bigl[\tfrac{1}{\kappa(T-R)}+\tfrac{1}{(T-R)^2}\bigr]$.

For the second claim, fix $\varepsilon>0$ and a closed interval $[\underline x,1]\ni\bar x$ with $0<\underline x<\bar x$. Split $[1,\infty)=[1,R_0]\cup[R_0,\infty)$. On $[R_0,\infty)$ the first part gives $|J_1^0(\rho,1;x_0)-J_1^0(\rho,1;\bar x)|\le \tfrac{2C_0}{R_0}+\tfrac{T-P}{\kappa}|x_0^2-\bar x^2|$; choosing $R_0$ with $2C_0/R_0<\varepsilon/3$ and then $x_0$ close enough to $\bar x$ bounds this by $2\varepsilon/3$. On the compact $[1,R_0]$, the trajectory, the hitting times and the integrand depend continuously on $(\rho,x_0)$, so $(\rho,x_0)\mapsto J_1^0(\rho,1;x_0)$ is jointly continuous, hence uniformly continuous on $[1,R_0]\times[\underline x,1]$; thus $\sup_{\rho\in[1,R_0]}|J_1^0(\rho,1;x_0)-J_1^0(\rho,1;\bar x)|<\varepsilon/3$ for $x_0$ close to $\bar x$. The claim follows.
\end{proof}

\subsection{Moving to the Discounted Game}\label{sec:discounted}

In this section we establish that the symmetric profile $\rho = 1$, which is the unique strict global maximiser of the undiscounted excess payoff, continues to determine a strict global symmetric Nash equilibrium for all discount rates $r > 0$, provided the initial condition is chosen appropriately.

The discounted payoff $V_1(\rho, 1; r) = r \int_0^\infty e^{-rt} u_1(x(t; \rho), y(t; \rho))\, dt$ can be written in terms of the discounted excess payoff
\begin{align*}
V_1(\rho, 1; r) = P + r\, W(\rho, r), \qquad W(\rho, r) = \int_0^{H'(\rho)} e^{-rt} \bigl( u_1(x(t; \rho), y(t; \rho)) - P \bigr)\, dt,
\end{align*}
where $H'(\rho) = \max(T_x(\rho), T_y(\rho))$ is the time at which both coordinates are absorbed at zero. For any $r > 0$, maximising $V_1$ is equivalent to maximising $W$, and as $r \to 0^+$, $W(\rho, r)$ converges to the undiscounted excess payoff $J_1^0(\rho, 1; x_0)$.

The scaling identity reduces the two-parameter problem to a single effective discount rate.

\begin{lemma}[Effective discount rate]
\label{lem:effective_discount}
For any learning rates $\gamma_1, \gamma_2 > 0$ and discount rate $r > 0$, with $\rho = \gamma_1 / \gamma_2$ and effective discount rate $\tilde{r} = r / \gamma_2$,
\begin{align*}
V_1(\gamma_1, \gamma_2; r; x_0) = V_1(\rho, 1; \tilde{r}; x_0) = P + \tilde{r}\, W(\rho, \tilde{r}; x_0).
\end{align*}
\end{lemma}

\begin{proof}
We have that $V_1(\lambda \gamma_1, \lambda \gamma_2; \lambda r; x_0) = V_1(\gamma_1, \gamma_2; r; x_0)$ for any $\lambda > 0$. Choosing $\lambda = 1/\gamma_2$ gives the first equality, and the second is the decomposition of $V_1(\rho, 1; \tilde{r})$ into its constant and excess parts.
\end{proof}

We now establish uniform convergence of $W$ to $J_1^0$, both on compact sets and on the upper half-line.

\begin{lemma}[Uniform convergence of $W$ on compact sets and the half-line]\label{lem:uniform-convergence-W}
Fix $x_0 = y_0 \in (0,1]$. For every compact interval $K = [\underline{\rho}, \overline{\rho}] \subset (0, \infty)$,
\begin{align*}
\sup_{\rho \in K} \bigl| W(\rho, r) - J_1^0(\rho, 1; x_0) \bigr| \;\xrightarrow{r \to 0^+}\; 0.
\end{align*}
Moreover, uniform convergence extends to the half-line $[1, \infty)$:
\begin{align*}
\sup_{\rho \ge 1} \bigl| W(\rho, r) - J_1^0(\rho, 1; x_0) \bigr| \;\xrightarrow{r \to 0^+}\; 0.
\end{align*}
\end{lemma}

\begin{proof}
Let $M := \sup_{(x,y) \in [0,1]^2} |u_1(x,y) - P| < \infty$. For any $\rho > 0$ and $r > 0$,
\begin{align*}
\bigl| W(\rho, r) - J_1^0(\rho, 1; x_0) \bigr|
&= \left| \int_0^{H'(\rho)} \bigl( e^{-rt} - 1 \bigr) (u_1(x(t;\rho),y(t;\rho)) - P)\, dt \right| \\
&\le M \int_0^{H'(\rho)} \bigl( 1 - e^{-rt} \bigr)\, dt \\
&\le M \cdot H'(\rho) \cdot \bigl( 1 - e^{-r H'(\rho)} \bigr),
\end{align*}
where the last bound uses $1 - e^{-rt} \le 1 - e^{-r H'(\rho)}$ for $t \in [0, H'(\rho)]$. It therefore suffices to show that $H'(\rho)$ is uniformly bounded on the relevant domain.

\textit{Bound on $T_y(\rho)$.} Since $\delta < 0$ and $x(t;\rho) \le x_0 \le 1$ on $[0, T_y(\rho)]$, we have 
\begin{align*}
\dot{y}(t;\rho) = -\frac{1}{2}\bigl( \kappa + \delta x(t;\rho) \bigr) \le -\frac{1}{2}(T - R),
\end{align*}
and integrating gives $T_y(\rho) \le \frac{2 y_0}{T - R} =: T_y^{\max},$ uniformly in $\rho > 0$.

\textit{Bound on $T_x(\rho)$.} On $[0, T_x(\rho)]$, $y(t;\rho) \le y_0 \le 1$ implies
\begin{align*}
\dot{x}(t;\rho) = -\frac{\rho}{2}\bigl( \kappa + \delta y(t;\rho) \bigr) \le -\frac{\rho}{2}(T - R) \le -\frac{\underline{\rho}}{2}(T - R),
\end{align*}
and integrating gives $T_x(\rho) \le \frac{2 x_0}{\underline{\rho}(T - R)} =: T_x^{\max},$ uniformly in $\rho \in K$.

After $\min(T_x(\rho), T_y(\rho))$, the trajectory enters the boundary phase. If $T_y(\rho) < T_x(\rho)$, then $y$ is absorbed and $\dot{x} = -\rho \kappa / 2$ thereafter, so the remaining time to absorb $x$ is at most $2 x_0 / (\underline{\rho} \kappa)$. Symmetrically, if $T_x(\rho) < T_y(\rho)$, the remaining time to absorb $y$ is at most $2 y_0 / \kappa$. Therefore,
\begin{align*}
   H'(\rho) \leq  \begin{cases}
        \frac{2y_0}{T-R} + \frac{2x_0}{\underline{\rho}(P - S)} & \text{ if } T_y(\rho) \leq T_x(\rho) \\
        \frac{2x_0}{\underline{\rho}(T - R)} + \frac{2y_0}{P-S} & \text{ if } T_y(\rho) > T_x(\rho)
    \end{cases}
\end{align*}
In either case, $H'(\rho) \leq H^{\max}$ where $H^{\max} \in \mathbb{R}$ depends only on $K$, $x_0$, $y_0$, and the game parameters. Combining,
\begin{align*}
\sup_{\rho \in K} \bigl| W(\rho, r) - J_1^0(\rho, 1; x_0) \bigr| \le M \cdot H^{\max} \cdot \bigl( 1 - e^{-r H^{\max}} \bigr) \xrightarrow{r \to 0^+} 0.
\end{align*}
For the half-line $[1, \infty)$, since $\rho \ge 1$ Firm~$1$'s coordinate is absorbed first, so the interior phase ends at $T_x(\rho) \le \tfrac{2x_0}{T-R}$ (taking $\underline{\rho} = 1$) and the subsequent slide, governed by $\dot{y} = -\kappa/2$, absorbs $y$ in time at most $\tfrac{2y_0}{\kappa} = \tfrac{2y_0}{P-S}$. Hence
\begin{align*}
H'(\rho) \leq \frac{2x_0}{T-R} + \frac{2 y_0}{P-S} =: H^{\max},
\end{align*}
and therefore
\begin{align*}
\sup_{\rho \ge 1} \bigl| W(\rho, r) - J_1^0(\rho, 1; x_0) \bigr| \le M \cdot H^{\max} \cdot \bigl( 1 - e^{-r H^{\max}} \bigr) \xrightarrow{r \to 0^+} 0,
\end{align*}
which completes the proof.
\end{proof}

\begin{remark}\label{rem:uniform-in-x0}
All constants in the proof of Lemma~\ref{lem:uniform-convergence-W} (namely $M$, the bounds $T_x^{\max}$, $T_y^{\max}$, and $H^{\max}$) depend on the game parameters and on $x_0=y_0$ only through continuous expressions. Hence for any $\bar x\in(0,1]$ they may be chosen uniformly for $x_0$ in a neighbourhood of $\bar x$, and the convergences stated in Lemma~\ref{lem:uniform-convergence-W} are locally uniform in $x_0=y_0$.
\end{remark}

To confine the maximiser away from the lower tail $\rho \to 0$, we use the following elementary bound.

\begin{lemma}\label{lem:g-lower-bound}
For all $z > 0$, the function $g(z) := \dfrac{z - 1 + e^{-z}}{z^2}$ satisfies
\begin{align*}
g(z) \ge \min\!\left(\frac{1}{4},\, \frac{1}{2z}\right).
\end{align*}
\end{lemma}

\begin{proof}
We divide the domain $(0, \infty)$ at the crossover point $z = 2$, where $\frac{1}{4} = \frac{1}{2z}$.

\textbf{Case 1: $z \ge 2$.} Dropping the exponential,
\begin{align*}
g(z) = \frac{z - 1 + e^{-z}}{z^2} > \frac{z - 1}{z^2}.
\end{align*}
For $z \ge 2$ we have $z - 1 \ge z/2$, so $g(z) > \frac{z/2}{z^2} = \frac{1}{2z}$.

\textbf{Case 2: $0 < z \le 2$.} We show $g$ is strictly decreasing on $(0,2]$. By the quotient rule,
\begin{align*}
g'(z) = \frac{(1 - e^{-z})z^2 - 2z(z - 1 + e^{-z})}{z^4} = \frac{2 - z - (2+z)e^{-z}}{z^3}.
\end{align*}
Let $\nu(z) = 2 - z - (2+z)e^{-z}$. Then
\begin{align*}
\nu(0) = 0, \qquad \nu'(z) = -1 + (1+z)e^{-z}, \qquad \nu'(0) = 0, \qquad \nu''(z) = -z e^{-z}.
\end{align*}
Since $\nu''(z) < 0$ for $z > 0$, $\nu'$ is strictly decreasing; with $\nu'(0) = 0$ this gives $\nu'(z) < 0$ and hence $\nu(z) < 0$ for $z > 0$. As $z^3 > 0$, we conclude $g'(z) < 0$, so $g$ is strictly decreasing. Therefore on $(0,2]$,
\begin{align*}
g(z) \ge g(2) = \frac{2 - 1 + e^{-2}}{4} = \frac{1}{4} + \frac{e^{-2}}{4} > \frac{1}{4}.
\end{align*}
Combining both cases gives $g(z) \ge \min\left(\frac{1}{4}, \frac{1}{2z}\right)$ for all $z > 0$.
\end{proof}

\begin{lemma}[Lower-tail penalty]\label{lem:lower-tail-penalty}
Fix $x_0 = y_0 \in (0, 1]$. There exist $\underline{\rho} \in (0, 1)$ and $\bar{r}_{\mathrm{low}} > 0$ such that
\begin{align*}
\sup_{\rho \in (0,\, \underline{\rho}]} W(\rho, r) \;<\; J_1^0(1, 1; x_0) - 1 \qquad \text{for all } r \in (0, \bar{r}_{\mathrm{low}}].
\end{align*}
\end{lemma}

\begin{proof}
Let $T_{\max} := 2 y_0 / (T - R)$ be the uniform bound on $T_y(\rho)$ from Lemma~\ref{lem:uniform-convergence-W}, and let $M := \sup_{(x,y) \in [0,1]^2} |u_1(x,y) - P| < \infty$.

For $\rho \le 1$, $T_y(\rho)$ is reached before $T_x(\rho)$. On $[T_y(\rho), T_x(\rho)]$, $y(t) = 0$ and $u_1(x(t), 0) - P = -\kappa x(t)$, so
\begin{align*}
W(\rho, r) = \int_0^{T_y(\rho)} e^{-rt}\bigl(u_1(x(t), y(t)) - P\bigr)\, dt \;-\; \kappa \int_{T_y(\rho)}^{T_x(\rho)} e^{-rt} x(t)\, dt,
\end{align*}
where the first integral is bounded above by $M T_{\max}$.

On $[T_y(\rho), T_x(\rho)]$, $\dot{x} = -\rho \kappa / 2$, so $x(t) = x(T_y(\rho)) - \tfrac{\rho \kappa}{2}(t - T_y(\rho))$ and $T_x(\rho) - T_y(\rho) = 2 x(T_y(\rho)) / (\rho \kappa)$. Evaluating the integral exactly,
\begin{align*}
\kappa \int_{T_y(\rho)}^{T_x(\rho)} e^{-rt} x(t)\, dt = e^{-r T_y(\rho)}\, \frac{2 x(T_y(\rho))^2}{\rho}\, g\!\left(\frac{2 x(T_y(\rho))\, r}{\rho \kappa}\right),
\end{align*}
with $g$ as in Lemma~\ref{lem:g-lower-bound}. Since $\dot{x}(t) \ge -\rho \kappa / 2$ on $[0, T_y(\rho)]$ and $T_y(\rho) \le T_{\max}$, we have $x(T_y(\rho)) \ge x_0 - \rho \kappa T_{\max} / 2$. Choosing $\underline{\rho}_0 := \min\bigl(1,\, x_0 / (\kappa T_{\max})\bigr)$ gives $x(T_y(\rho)) \ge x_0 / 2$ for all $\rho \in (0, \underline{\rho}_0]$, and Lemma~\ref{lem:g-lower-bound} yields
\begin{align*}
\kappa \int_{T_y(\rho)}^{T_x(\rho)} e^{-rt} x(t)\, dt
\ge e^{-r T_{\max}}\, \frac{2 x(T_y(\rho))^2}{\rho} \cdot \min\!\left(\frac{1}{4},\, \frac{\rho \kappa}{4 x(T_y(\rho))\, r}\right)
\ge e^{-r T_{\max}}\, \min\!\left(\frac{x_0^2}{8 \rho},\, \frac{x_0\, \kappa}{4 r}\right).
\end{align*}
Hence, for $\rho \in (0, \underline{\rho}_0]$,
\begin{align*}
W(\rho, r) \;\le\; M T_{\max} \;-\; e^{-r T_{\max}}\, \min\!\left(\frac{x_0^2}{8 \rho},\, \frac{x_0\, \kappa}{4 r}\right).
\end{align*}
Pick $C := M T_{\max} + |J_1^0(1,1;x_0)| + 2$. Choose $\underline{\rho} \in (0, \underline{\rho}_0]$ with $x_0^2 / (8 \underline{\rho}) \ge 2 C$, and $\bar{r}_{\mathrm{low}} \in (0, 1]$ with $x_0\, \kappa / (4 \bar{r}_{\mathrm{low}}) \ge 2 C$ and $e^{-\bar{r}_{\mathrm{low}} T_{\max}} \ge 1/2$. Then for all $\rho \in (0, \underline{\rho}]$ and $r \in (0, \bar{r}_{\mathrm{low}}]$,
\begin{align*}
e^{-r T_{\max}}\, \min\!\left(\frac{x_0^2}{8 \rho},\, \frac{x_0\, \kappa}{4 r}\right) \;\ge\; \frac{1}{2} \cdot 2 C \;=\; C,
\end{align*}
and so $W(\rho, r) \le M T_{\max} - C \le -|J_1^0(1,1;x_0)| - 2 < J_1^0(1,1;x_0) - 1$.
\end{proof}

One final technical result is required: 

\begin{lemma}\label{lem:final}
Fix $\delta<0$ and $\bar x\in(0,1]$. There exist a neighbourhood $U'$ of $\rho=1$ in $(0,\infty)$, a neighbourhood $X'$ of $\bar x$ in $(0,1]$, and a constant $\bar r>0$ such that the following hold.
\begin{enumerate}
    \item[(i)] For every $(\tilde r,x_0)\in[0,\bar r)\times X'$ the map $\rho\mapsto W(\rho,\tilde r;x_0)$ is twice continuously differentiable on $U'$, and the maps $(\rho,\tilde r,x_0)\mapsto\partial_\rho W(\rho,\tilde r;x_0)$ and $(\rho,\tilde r,x_0)\mapsto\partial_{\rho\rho}W(\rho,\tilde r;x_0)$ are jointly continuous on $U'\times[0,\bar r)\times X'$.
    \item[(ii)] The map $(\tilde r,x_0)\mapsto\partial_\rho W(1,\tilde r;x_0)$ is continuously differentiable on $[0,\bar r)\times X'$, with
    \begin{align*}
        \partial_{\tilde r}\bigl[\partial_\rho W(1,\tilde r;x_0)\bigr]=-\int_0^{T^*(x_0)}t\,e^{-\tilde r t}\,\partial_\rho u_1(x(t;1,x_0),y(t;1,x_0))\,dt,
    \end{align*}
    where $T^*(x_0)=\tfrac{2}{\delta}\ln\!\bigl(1+\tfrac{\delta x_0}{\kappa}\bigr)$.
    \item[(iii)] For every $\rho\in U'$ and $x_0\in X'$, $\partial_{\rho\rho}W(\rho,0;x_0)=\partial_{\rho\rho}J_1^0(\rho,1;x_0)$.
\end{enumerate}
\end{lemma}

\begin{proof}[Proof of Lemma~\ref{lem:final}]
We first define the interior trajectories without boundary projection. Let $x(t; \rho, x_0)$ and $y(t; \rho, x_0)$ be the unique solutions to the unprojected ODE system:
\begin{align*}
    \dot{x} &= -\frac{\rho}{2}(\kappa + \delta y), &
    \dot{y} = -\frac{1}{2}(\kappa + \delta x),
\end{align*}
with initial conditions $x(0) = y(0) = x_0$. The vector field is smooth and the initial datum enters smoothly, so $(t,\rho,x_0)\mapsto x(t;\rho,x_0)$ and $(t,\rho,x_0)\mapsto y(t;\rho,x_0)$ are jointly smooth. At $\rho=1$ the symmetry of the initial condition gives $x(t;1,x_0)=y(t;1,x_0)$ for all $t\ge 0$, and this common coordinate reaches zero at $T^*(x_0)=\tfrac{2}{\delta}\ln(1+\delta x_0/\kappa)>0$. Write $u(t;x_0):=\partial_\rho x(t;1,x_0)$ and $v(t;x_0):=\partial_\rho y(t;1,x_0)$ for the first-order sensitivities at the symmetric profile.

Let $T_x(\rho,x_0)$ and $T_y(\rho,x_0)$ solve $x(T_x;\rho,x_0)=0$ and $y(T_y;\rho,x_0)=0$. Since $\dot x(T^*(\bar x);1,\bar x)=\dot y(T^*(\bar x);1,\bar x)=-\kappa/2\neq 0$, the implicit function theorem provides neighbourhoods $U'$ of $\rho=1$ in $(0,\infty)$ and $X'$ of $\bar x$ in $(0,1]$ on which $T_x$ and $T_y$ are well-defined and jointly smooth in $(\rho,x_0)$. As $T_x(1,x_0)=T_y(1,x_0)=T^*(x_0)$ for every $x_0\in X'$, differentiating the defining identities in $\rho$ at $\rho=1$ gives:
\begin{align*}
\partial_\rho T_x(1,x_0)=\frac{2}{\kappa}u(T^*(x_0);x_0), \qquad \partial_\rho T_y(1,x_0)=\frac{2}{\kappa}v(T^*(x_0);x_0).
\end{align*}
The computation in the proof of Theorem~\ref{thm:undiscounted-FOC} gives $u(T^*;x_0)-v(T^*;x_0)=-(K_0/\delta)\sinh(\delta T^*/2)$ with $K_0=\kappa+\delta x_0$, which is strictly negative for every $x_0\in(0,1]$ since $\delta<0$. Hence $\partial_\rho T_x(1,x_0)<\partial_\rho T_y(1,x_0)$ throughout $X'$, and after shrinking $U'$ we have $T_x(\rho,x_0)<T_y(\rho,x_0)$ for $\rho>1$ and $T_y(\rho,x_0)<T_x(\rho,x_0)$ for $\rho<1$, uniformly in $x_0\in X'$.

We can now express $W(\rho, \tilde{r};x_0)$ piecewise on $U'$. For $\rho\ge 1$ Firm~$1$'s coordinate is absorbed first; thereafter $x\equiv 0$ and $y$ decays at rate $\kappa/2$ until it vanishes at $T_2^+(\rho,x_0)=T_x(\rho,x_0)+\tfrac{2}{\kappa}y(T_x(\rho,x_0);\rho,x_0)$. Thus,
\begin{align*}
W_+(\rho,\tilde r;x_0)=\int_0^{T_x(\rho,x_0)} e^{-\tilde r t}F_1(t,\rho;x_0)\,dt+\int_{T_x(\rho,x_0)}^{T_2^+(\rho,x_0)} e^{-\tilde r t}F_2(t,\rho;x_0)\,dt,
\end{align*}
where $F_1(t,\rho;x_0)=(S-P)x(t;\rho,x_0)+(T-P)y(t;\rho,x_0)-\delta x(t;\rho,x_0)y(t;\rho,x_0)$ and $F_2(t,\rho;x_0)=(T-P)\bigl[y(T_x(\rho,x_0);\rho,x_0)-\tfrac{\kappa}{2}(t-T_x(\rho,x_0))\bigr]$. 

For $\rho \le 1$, Firm~$2$'s coordinate is absorbed first, after which $x$ decays at rate $\rho\kappa/2$ until $T_2^-(\rho,x_0)=T_y(\rho,x_0)+\tfrac{2}{\rho\kappa}x(T_y(\rho,x_0);\rho,x_0)$, giving
\begin{align*}
W_-(\rho,\tilde r;x_0)=\int_0^{T_y(\rho,x_0)} e^{-\tilde r t}F_1(t,\rho;x_0)\,dt+\int_{T_y(\rho,x_0)}^{T_2^-(\rho,x_0)} e^{-\tilde r t}F_3(t,\rho;x_0)\,dt,
\end{align*}
with $F_3(t,\rho;x_0)=(S-P)\bigl[x(T_y(\rho,x_0);\rho,x_0)-\tfrac{\rho\kappa}{2}(t-T_y(\rho,x_0))\bigr]$. All endpoints and integrands are jointly smooth in $(\rho,x_0)$ on $U'\times X'$.

Since all bounding functions and integrands are smooth on $U'$, we can repeatedly apply Leibniz's rule. For $\rho > 1$:
\begin{align*}
    \partial_\rho W_+(\rho,\tilde r;x_0) &= \int_0^{T_x(\rho, x_0)} e^{-\tilde{r}t} \partial_\rho F_1(t,\rho;x_0) \, dt + e^{-\tilde{r}T_x(\rho,x_0)} F_1(T_x(\rho), \rho; x_0)\, \partial_\rho T_x(\rho,x_0) \\
    &\quad + \int_{T_x(\rho,x_0)}^{T_2^+(\rho,x_0)} e^{-\tilde{r}t} \partial_\rho F_2(t,\rho;x_0) \, dt + e^{-\tilde{r}T_2^+(\rho,x_0)} F_2(T_2^+(\rho,x_0), \rho;x_0)\, \partial_\rho T_2^+(\rho,x_0) \\
    &\quad - e^{-\tilde{r}T_x(\rho,x_0)} F_2(T_x(\rho,x_0), \rho;x_0)\, \partial_\rho T_x(\rho,x_0).
\end{align*}
Because $x(T_x;\rho,x_0)=0$ we have $F_1(T_x,\rho;x_0)=(T-P)y(T_x;\rho,x_0)=F_2(T_x,\rho;x_0)$, so the two boundary terms at $T_x$ cancel, and $F_2(T_2^+,\rho;x_0)=0$ by definition of $T_2^+$. Hence,
\begin{align}\label{eq:dW-plus}
\partial_\rho W_+(\rho,\tilde r;x_0)=\int_0^{T_x(\rho,x_0)} e^{-\tilde r t}\partial_\rho F_1(t,\rho;x_0)\,dt+\int_{T_x(\rho,x_0)}^{T_2^+(\rho,x_0)} e^{-\tilde r t}\partial_\rho F_2(t,\rho;x_0)\,dt.
\end{align}
An analogous computation for $\rho<1$, using $F_1(T_y;\rho,x_0)=(S-P)x(T_y;\rho,x_0)=F_3(T_y;\rho,x_0)$ and $F_3(T_2^-;\rho,x_0)=0$, gives
\begin{align}\label{eq:dW-minus}
\partial_\rho W_-(\rho,\tilde r;x_0)=\int_0^{T_y(\rho,x_0)} e^{-\tilde r t}\partial_\rho F_1(t,\rho;x_0)\,dt+\int_{T_y(\rho,x_0)}^{T_2^-(\rho,x_0)} e^{-\tilde r t}\partial_\rho F_3(t,\rho;x_0)\,dt.
\end{align}

We first verify that $\rho \mapsto \partial_\rho W(\rho,\tilde r;x_0)$ is continuous at $\rho=1$. As $\rho \to 1^{\pm}$ we have $T_x(\rho,x_0), T_y(\rho,x_0) \to T^*$ and $T_2^{\pm}(\rho,x_0) \to T^*$, so the second integrals in \eqref{eq:dW-plus} and \eqref{eq:dW-minus} are taken over intervals whose length shrinks to zero; since their integrands are bounded on $U'$, these integrals vanish in the limit. The first integrals both converge to $\int_0^{T^*(x_0)} e^{-\tilde r t}\partial_\rho F_1(t,1;x_0)\,dt$. Therefore
\begin{align}\label{eq:dW-at-one}
\lim_{\rho\to 1^+}\partial_\rho W_+(\rho,\tilde r;x_0)=\int_0^{T^*(x_0)} e^{-\tilde r t}\partial_\rho F_1(t,1;x_0)\,dt=\lim_{\rho\to 1^-}\partial_\rho W_-(\rho,\tilde r;x_0),
\end{align}
so $\rho\mapsto\partial_\rho W(\rho,\tilde r;x_0)$ extends continuously across $\rho=1$, with value at $\rho=1$ given by \eqref{eq:dW-at-one}.

Differentiating \eqref{eq:dW-plus} a second time yields
\begin{align*}
\partial_{\rho\rho} W_+(\rho,\tilde r;x_0) &= \int_0^{T_x(\rho,x_0)} e^{-\tilde r t}\,\partial_{\rho\rho} F_1(t,\rho;x_0)\,dt + e^{-\tilde r T_x(\rho,x_0)}\,\partial_\rho F_1(T_x(\rho,x_0),\rho;x_0)\,\partial_\rho T_x(\rho,x_0) \\
&\quad + \int_{T_x(\rho,x_0)}^{T_2^+(\rho,x_0)} e^{-\tilde r t}\,\partial_{\rho\rho} F_2(t,\rho;x_0)\,dt + e^{-\tilde r T_2^+(\rho,x_0)}\,\partial_\rho F_2(T_2^+(\rho,x_0),\rho;x_0)\,\partial_\rho T_2^+(\rho,x_0) \\
&\quad - e^{-\tilde r T_x(\rho,x_0)}\,\partial_\rho F_2(T_x(\rho,x_0),\rho;x_0)\,\partial_\rho T_x(\rho,x_0).
\end{align*}
Taking the limit as $\rho \to 1^+$ (so $T_x(\rho,x_0) \to T^*$ and $T_2^+(\rho,x_0) \to T^*$), the integral over $[T_x(\rho), T_2^+(\rho)]$ vanishes. A direct computation gives
\begin{align*}
    \partial_\rho F_2(t, \rho) = (T-P)\Bigl[ \partial_\rho y(T_x(\rho,x_0);\rho,x_0)+\partial_t y(T_x(\rho,x_0);\rho,x_0)\,\partial_\rho T_x(\rho,x_0)+\frac{\kappa}{2}\,\partial_\rho T_x(\rho,x_0) \Bigr],
\end{align*}
which is independent of $t$. At $\rho=1$, $\partial_t y(T^*(x_0);1,x_0)=-\kappa/2$ cancels the last two bracketed terms, so $\partial_\rho F_2(t,1;x_0)=(T-P)\,\partial_\rho y(T^*(x_0);1,x_0)$, and
\begin{align*}
\partial_\rho T_2^+(1,x_0)=\partial_\rho T_x(1,x_0)+\frac{2}{\kappa}\left[\partial_t y(T^*(x_0);1,x_0)\,\partial_\rho T_x(1,x_0)+\partial_\rho y(T^*(x_0);1,x_0)\right]=\partial_\rho T_y(1,x_0).
\end{align*}
The right limit of the second derivative is therefore
\begin{align*}
\lim_{\rho\to 1^+}\partial_{\rho\rho} W_+(\rho,\tilde r;x_0) &= \int_0^{T^*(x_0)} e^{-\tilde r t}\,\partial_{\rho\rho} F_1(t,1;x_0)\,dt \\
& \quad + e^{-\tilde r T^*(x_0)}\bigl [(S-P)\,\partial_\rho x(T^*(x_0);1,x_0)\,\partial_\rho T_x(1,x_0) \\
& \qquad +(T-P)\,\partial_\rho y(T^*(x_0);1,x_0)\,\partial_\rho T_y(1,x_0)\bigr ].
\end{align*}

A completely analogous calculation for $\rho < 1$, differentiating \eqref{eq:dW-minus}, gives
\begin{align*}
\partial_{\rho\rho} W_-(\rho,\tilde r;x_0) &= \int_0^{T_y(\rho,x_0)} e^{-\tilde r t}\,\partial_{\rho\rho} F_1(t,\rho;x_0)\,dt + e^{-\tilde r T_y(\rho,x_0)}\,\partial_\rho F_1(T_y(\rho,x_0),\rho;x_0)\,\partial_\rho T_y(\rho,x_0) \\
&\quad + \int_{T_y(\rho,x_0)}^{T_2^-(\rho,x_0)} e^{-\tilde r t}\,\partial_{\rho\rho} F_3(t,\rho;x_0)\,dt + e^{-\tilde r T_2^-(\rho,x_0)}\,\partial_\rho F_3(T_2^-(\rho,x_0),\rho;x_0)\,\partial_\rho T_2^-(\rho,x_0) \\
&\quad - e^{-\tilde r T_y(\rho,x_0)}\,\partial_\rho F_3(T_y(\rho,x_0),\rho;x_0)\,\partial_\rho T_y(\rho,x_0).
\end{align*}
The quantity
\begin{align*}
\partial_\rho F_3(t,\rho;x_0) & =(S-P)\bigl[\partial_\rho x(T_y(\rho,x_0);\rho,x_0)+\partial_t x(T_y(\rho,x_0);\rho,x_0)\,\partial_\rho T_y(\rho,x_0) \\
& \quad -\frac{\kappa}{2}\bigl(t-T_y(\rho,x_0)\bigr)+\frac{\rho\kappa}{2}\,\partial_\rho T_y(\rho,x_0)\bigr]
\end{align*}
at $\rho=1$, $t=T^*(x_0)$, using $\partial_t x(T^*(x_0);1,x_0)=-\kappa/2$, reduces to $\partial_\rho F_3(T^*(x_0),1;x_0)=(S-P)\,\partial_\rho x(T^*(x_0);1,x_0)$.  The factor $x(T_y(\rho,x_0);\rho,x_0)$ arising when differentiating $T_2^-(\rho,x_0)$ vanishes at $\rho=1$, since $x(T_y(1,x_0);1,x_0)=0$, so
\begin{align*}
\partial_\rho T_2^-(1,x_0)=\partial_\rho T_y(1,x_0)+\frac{2}{\kappa}\left[\partial_t x(T^*(x_0);1,x_0)\,\partial_\rho T_y(1,x_0)+\partial_\rho x(T^*(x_0);1,x_0)\right]=\partial_\rho T_x(1,x_0).
\end{align*}

Letting $\rho\to 1^-$, the integral over $[T_y(\rho,x_0),T_2^-(\rho,x_0)]$ vanishes, and using $T_y(1,x_0)=T_2^-(1,x_0)=T^*(x_0)$,
\begin{align*}
\lim_{\rho\to 1^-}\partial_{\rho\rho} W_-(\rho,\tilde r;x_0) &= \int_0^{T^*(x_0)} e^{-\tilde r t}\,\partial_{\rho\rho} F_1(t,1;x_0)\,dt \\
&\quad + e^{-\tilde r T^*(x_0)}\Bigl[\partial_\rho F_1(T^*(x_0),1;x_0)\,\partial_\rho T_y(1,x_0) \\
&\qquad\qquad + \partial_\rho F_3(T^*(x_0),1;x_0)\,\partial_\rho T_2^-(1,x_0) - \partial_\rho F_3(T^*(x_0),1;x_0)\,\partial_\rho T_y(1,x_0)\Bigr].
\end{align*}
Substituting $\partial_\rho F_3(T^*(x_0),1;x_0)=(S-P)\,\partial_\rho x(T^*(x_0);1,x_0)$, $\partial_\rho T_2^-(1,x_0)=\partial_\rho T_x(1,x_0)$, and $\partial_\rho F_1(T^*(x_0),1;x_0)=(S-P)\,\partial_\rho x(T^*(x_0);1,x_0)+(T-P)\,\partial_\rho y(T^*(x_0);1,x_0)$, and simplifying,
\begin{align*}
\lim_{\rho\to 1^-}\partial_{\rho\rho} W_-(\rho,\tilde r;x_0) &= \int_0^{T^*(x_0)} e^{-\tilde r t}\,\partial_{\rho\rho} F_1(t,1;x_0)\,dt \\
&\quad + e^{-\tilde r T^*(x_0)}\bigl [(S-P)\,\partial_\rho x(T^*(x_0);1,x_0)\,\partial_\rho T_x(1,x_0) \\
& \qquad +(T-P)\,\partial_\rho y(T^*(x_0);1,x_0)\,\partial_\rho T_y(1,x_0)\bigr ],
\end{align*}
which coincides with the right limit. Thus the one-sided second derivatives agree at $\rho=1$. We have already established that $\rho \mapsto \partial_\rho W(\rho,\tilde r;x_0)$ is continuous on $U'$; it is moreover continuously differentiable on $U' \setminus \{1\}$, and by the above its derivative $\partial_{\rho\rho} W(\rho,\tilde r;x_0)$ admits a common limit as $\rho \to 1$. Hence, $\rho \mapsto \partial_\rho W(\rho,\tilde r;x_0)$ is therefore differentiable at $\rho=1$, with $\partial_{\rho\rho} W(1, \tilde{r};x_0)$ equal to this common limit. Thus $\rho\mapsto W(\rho,\tilde r;x_0)$ is twice continuously differentiable on $U'$ for every $(\tilde r,x_0)\in[0,\bar r)\times X'$,

For joint continuity, regard the second-derivative expression obtained from \eqref{eq:dW-plus} as a function on the closed half $\{\rho\ge 1\}\times[0,\bar r)\times X'$. Each term is either an integral of a jointly continuous integrand against jointly smooth limits, or a product of the jointly continuous quantities $e^{-\tilde r T_x(\rho,x_0)}$, $\partial_\rho F_1(T_x(\rho,x_0),\rho;x_0)$, $\partial_\rho T_x(\rho,x_0)$ and their $T_2^+$-analogues, so $\partial_{\rho\rho} W_+$ is jointly continuous there, the integral over $[T_x(\rho,x_0),T_2^+(\rho,x_0)]$ vanishing at $\rho=1$. Likewise $\partial_{\rho\rho} W_-$ is jointly continuous on $\{\rho\le 1\}\times[0,\bar r)\times X'$. As the two agree on $\rho=1$, the glued function $\partial_{\rho\rho} W$ is jointly continuous on $U'\times[0,\bar r)\times X'$; the same argument applied to \eqref{eq:dW-plus} and \eqref{eq:dW-minus} themselves gives joint continuity of $\partial_\rho W$. This establishes~(i).

By \eqref{eq:dW-at-one}, for every $(\tilde r,x_0)\in[0,\bar r)\times X'$ we have $\partial_\rho W(1,\tilde r;x_0)=\int_0^{T^*(x_0)} e^{-\tilde r t}\,\partial_\rho F_1(t,1;x_0)\,dt$, where
\begin{align*}
\partial_\rho F_1(t,1;x_0)&=(S-P)\,\partial_\rho x(t;1,x_0)+(T-P)\,\partial_\rho y(t;1,x_0)-\delta\bigl(\partial_\rho x(t;1,x_0)+\partial_\rho y(t;1,x_0)\bigr)\,x(t;1,x_0)\\
&=\partial_\rho u_1\bigl(x(t;\rho,x_0),y(t;\rho,x_0)\bigr)\big|_{\rho=1}.
\end{align*}
Since the upper limit $T^*(x_0)$ does not depend on $\tilde r$, differentiation under the integral sign in $\tilde r$ is justified and yields
\begin{align*}
\partial_{\tilde r}\bigl[\partial_\rho W(1,\tilde r;x_0)\bigr]=-\int_0^{T^*(x_0)} t\,e^{-\tilde r t}\,\partial_\rho u_1\bigl(x(t;\rho,x_0),y(t;\rho,x_0)\bigr)\big|_{\rho=1}\,dt,
\end{align*}
which is jointly continuous in $(\tilde r,x_0)$. Differentiating instead in $x_0$, Leibniz's rule gives
\begin{align*}
\partial_{x_0}\bigl[\partial_\rho W(1,\tilde r;x_0)\bigr]=\int_0^{T^*(x_0)} e^{-\tilde r t}\,\partial_{x_0}\partial_\rho F_1(t,1;x_0)\,dt+e^{-\tilde r T^*(x_0)}\,\partial_\rho F_1\bigl(T^*(x_0),1;x_0\bigr)\,\partial_{x_0} T^*(x_0),
\end{align*}
likewise jointly continuous in $(\tilde r,x_0)$. As both first-order partials exist and are jointly continuous, $(\tilde r,x_0)\mapsto\partial_\rho W(1,\tilde r;x_0)$ is continuously differentiable on $[0,\bar r)\times X'$, proving~(ii).

Finally, at $\tilde r=0$, we have $W(\rho,0;x_0)=\int_0^{H'(\rho,x_0)}\bigl(u_1(x(t;\rho,x_0),y(t;\rho,x_0))-P\bigr)\,dt\equiv J_1^0(\rho,1;x_0)$, where $H'(\rho,x_0)=\max(T_x(\rho,x_0),T_y(\rho,x_0))$. As these are identical functions of $\rho$ on $U'$ for each $x_0\in X'$, their second $\rho$-derivatives coincide, giving $\partial_{\rho\rho} W(\rho,0;x_0)=\partial_{\rho\rho} J_1^0(\rho,1;x_0)$ on $U'$. This establishes~(iii) and completes the proof.
\end{proof}

Next, we establish how a transversality condition required in an application of an implicit function theorem is generically satisfied. We prove something slightly stronger, which will become useful in the senstivity analysis of Proposition~\ref{prop:delta_sensitivity}

\begin{lemma}[Closed form of the transversality constant]\label{lem:transversality-constant}
Fix $\delta < 0$ and $x_0 = y_0 \in (0,1]$, and write $K_0 := \kappa + \delta x_0 > 0$ and $L := \ln(K_0/\kappa)$. Then we determine the following closed-form solution:
\begin{align}\label{eq:transversality-closed}
\partial^2_{\rho\tilde r}W(1,0;x_0)
= \frac{1}{\delta^3\kappa}\Bigl[
&(T-S)\bigl(K_0^2(1-L) + 4K_0\kappa - (2L^2 + 5L + 5)\kappa^2\bigr) \nonumber \\
&\quad - K_0^2\kappa + (2L^2 + 2L + 1)\kappa^3
\Bigr].
\end{align}
\end{lemma}

\begin{proof}
By Lemma~\ref{lem:final}(ii) evaluated at $\tilde r = 0$, together with Lemma~\ref{lem:final}(iii) and the identity $\partial_\rho F_1(t,1;x_0) = -K(t)\,u(t) + (T-S-K(t))\,v(t)$ established in the proof of Theorem~\ref{thm:undiscounted-FOC}, the transversality constant equals the first-order integrand of~\eqref{eq:FOC-undiscounted} weighted by $t$:
\begin{align}\label{eq:trans-defining}
\partial^2_{\rho\tilde r}W(1,0;x_0)
= -\int_0^{T^*} t\,\bigl[-K(t)\,u(t) + (T-S-K(t))\,v(t)\bigr]\,dt.
\end{align}
Recall from the proof of Theorem~\ref{thm:undiscounted-FOC} the substitution $p := u + v$ and $w := u - v$, equivalently $u = (p+w)/2$ and $v = (p-w)/2$. Substituting these into the integrand of~\eqref{eq:trans-defining},
\begin{align*}
-K u + (T-S-K) v
&= -K\cdot\tfrac{p+w}{2} + (T-S-K)\cdot\tfrac{p-w}{2} \\
&= \tfrac{T-S}{2}(p - w) - K\,\tfrac{(p+w)+(p-w)}{2} \\
&= \tfrac{T-S}{2}(p - w) - K p,
\end{align*}
so that
\begin{align}\label{eq:trans-pw}
\partial^2_{\rho\tilde r}W(1,0;x_0)
= -\int_0^{T^*} t\,\Bigl[\tfrac{T-S}{2}\bigl(p(t)-w(t)\bigr) - K(t)\,p(t)\Bigr]\,dt.
\end{align}

The functions $K$, $p$, and $w$ are given explicitly in the proof of Theorem~\ref{thm:undiscounted-FOC} by
\begin{align*}
K(t) = K_0\, e^{-\delta t/2}, \qquad
p(t) = -\tfrac{K_0 t}{2}\, e^{-\delta t/2}, \qquad
w(t) = -\tfrac{K_0}{\delta}\,\sinh\!\Bigl(\tfrac{\delta t}{2}\Bigr),
\end{align*}
and the hitting time satisfies $K(T^*) = \kappa$, whence
\begin{align*}
    e^{-\delta T^*/2} = \frac{\kappa}{K_0}, \qquad
T^* = \frac{2L}{\delta} = \frac{2}{\delta}\ln\!\Bigl(\frac{K_0}{\kappa}\Bigr).
\end{align*}

Substituting these expressions, together with $Kp = K_0 e^{-\delta t/2}\cdot\bigl(-\tfrac{K_0 t}{2}e^{-\delta t/2}\bigr) = -\tfrac{K_0^2 t}{2}e^{-\delta t}$, into~\eqref{eq:trans-pw} and expanding the three contributions yields
\begin{align}\label{eq:trans-three-integrals}
\partial^2_{\rho\tilde r}W(1,0;x_0)
= \frac{(T-S)K_0}{4}\,\mathcal{J}_1
\;-\; \frac{(T-S)K_0}{2\delta}\,\mathcal{J}_2
\;-\; \frac{K_0^2}{2}\,\mathcal{J}_3,
\end{align}
where
\begin{align*}
\mathcal{J}_1 := \int_0^{T^*}\!\! t^2 e^{-\delta t/2}\,dt, \qquad
\mathcal{J}_2 := \int_0^{T^*}\!\! t\,\sinh\!\Bigl(\tfrac{\delta t}{2}\Bigr)\,dt, \qquad
\mathcal{J}_3 := \int_0^{T^*}\!\! t^2 e^{-\delta t}\,dt.
\end{align*}

Each integral is computed by integration by parts. For $\mathcal{J}_1$, two integrations by parts give $\int t^2 e^{at}\,dt = e^{at}\bigl(\tfrac{t^2}{a} - \tfrac{2t}{a^2} + \tfrac{2}{a^3}\bigr)$ with $a = -\delta/2$; evaluating between $0$ and $T^*$ and substituting $e^{-\delta T^*/2} = \kappa/K_0$ and $T^* = 2L/\delta$,
\begin{align*}
\mathcal{J}_1 = \frac{8\bigl(2K_0 - \kappa(L^2 + 2L + 2)\bigr)}{\delta^3 K_0}.
\end{align*}
For $\mathcal{J}_2$, integration by parts with $\int t\,\sinh(bt)\,dt = \tfrac{t}{b}\cosh(bt) - \tfrac{1}{b^2}\sinh(bt)$ for $b = \delta/2$, evaluated at $T^*$ using $\cosh L = \tfrac12(K_0/\kappa + \kappa/K_0)$ and $\sinh L = \tfrac12(K_0/\kappa - \kappa/K_0)$,
\begin{align*}
\mathcal{J}_2 = \frac{2\bigl(-K_0^2 + L(K_0^2 + \kappa^2) + \kappa^2\bigr)}{\delta^2 K_0\,\kappa}.
\end{align*}
For $\mathcal{J}_3$, the same antiderivative as $\mathcal{J}_1$ with $a = -\delta$, using $e^{-\delta T^*} = (e^{-\delta T^*/2})^2 = \kappa^2/K_0^2$,
\begin{align*}
\mathcal{J}_3 = \frac{2\bigl(K_0^2 - \kappa^2(2L^2 + 2L + 1)\bigr)}{\delta^3 K_0^2}.
\end{align*}

Substituting $\mathcal{J}_1$, $\mathcal{J}_2$, $\mathcal{J}_3$ into~\eqref{eq:trans-three-integrals}, the three contributions become
\begin{align*}
\frac{(T-S)K_0}{4}\,\mathcal{J}_1 &= \frac{2(T-S)\bigl(2K_0 - \kappa(L^2+2L+2)\bigr)}{\delta^3}, \\
-\frac{(T-S)K_0}{2\delta}\,\mathcal{J}_2 &= -\frac{(T-S)\bigl(-K_0^2 + L(K_0^2+\kappa^2) + \kappa^2\bigr)}{\delta^3\kappa}, \\
-\frac{K_0^2}{2}\,\mathcal{J}_3 &= -\frac{K_0^2 - \kappa^2(2L^2+2L+1)}{\delta^3}.
\end{align*}
Placing all three over the common denominator $\delta^3\kappa$ and collecting the coefficient of $(T-S)$ separately from the remaining terms gives
\begin{align*}
\partial^2_{\rho\tilde r}W(1,0;x_0)
= \frac{1}{\delta^3\kappa}\Bigl[
&(T-S)\bigl(K_0^2(1-L) + 4K_0\kappa - (2L^2 + 5L + 5)\kappa^2\bigr) \\
&\quad - K_0^2\kappa + (2L^2 + 2L + 1)\kappa^3
\Bigr],
\end{align*}
which is~\eqref{eq:transversality-closed}.
\end{proof}

Next, we  an auxiliary Lemma, which will be used in the final Lemma that applies the closed form expression above and shows how the sign of the transversality constant depends on a payoff ratio.

\begin{lemma}\label{lem:D_unique_root}
The  function defined on $\bar{q} \in (0, 1)$ by:
\begin{equation*}
    D(\bar{q}) = \bar{q}^4\ln\bar{q} - 2\bar{q}^4 + 2\bar{q}^2\ln^2\bar{q} + 6\bar{q}^2\ln\bar{q} + 4\bar{q}^2 - 8\bar{q}\ln^2\bar{q} - 2\ln^2\bar{q} - 7\ln\bar{q} - 2
\end{equation*}
possesses exactly one root $\bar{q}_c \in (0, 1)$.
\end{lemma}

\begin{proof}
We deploy a cascade application of Rolle's Theorem paired with the Intermediate Value Theorem. 

\vspace{0.2cm}
\noindent \textbf{Step 1: Boundary Conditions and Asymptotics.} \\
We evaluate the limits at the endpoints of $(0, 1)$. At $\bar{q} \to 1^-$, straightforward evaluation confirms that the function and its first five derivatives vanish identically:
\begin{equation*}
    D(1) = D'(1) = D''(1) = D'''(1) = D^{(4)}(1) = D^{(5)}(1) = 0.
\end{equation*}
As $\bar{q} \to 0^+$, the asymptotic behavior is strictly dominated by the poles of the logarithmic terms $\mathcal{O}(\bar{q}^{-k} \ln \bar{q})$. Evaluating the dominant terms yields an alternating sign sequence:
\begin{align*}
    D(\bar{q}) &\sim -2\ln^2 \bar{q} \to -\infty \quad (-);&
    D'(\bar{q}) &\sim -4\frac{\ln \bar{q}}{\bar{q}} \to +\infty \quad (+) \\
    D''(\bar{q}) &\sim 4\frac{\ln \bar{q}}{\bar{q}^2} \to -\infty \quad (-); &   D'''(\bar{q}) &\sim -8\frac{\ln \bar{q}}{\bar{q}^3} \to +\infty \quad (+) \\
    D^{(4)}(\bar{q}) &\sim 24\frac{\ln \bar{q}}{\bar{q}^4} \to -\infty \quad (-);&
    D^{(5)}(\bar{q}) &\sim -96\frac{\ln \bar{q}}{\bar{q}^5} \to +\infty \quad (+)
\end{align*}

\vspace{0.2cm}
\noindent \textbf{Step 2: Isolating the Fifth Derivative.} \\
We compute the exact algebraic form of the fifth derivative:
\begin{equation*}
    D^{(5)}(\bar{q}) = \frac{8}{\bar{q}^5} f(\bar{q}),
\end{equation*}
where $f(\bar{q}) = 3\bar{q}^4 + \bar{q}^2(2\ln \bar{q} + 3) + 2\bar{q}(6\ln \bar{q} - 5) - 12\ln \bar{q} + 4$. 
Because $\bar{q} > 0$, the roots and signs of $D^{(5)}(\bar{q})$ are identical to those of $f(\bar{q})$. To study the roots of $f(\bar{q})$, we define the auxiliary function $g(\bar{q}) = \bar{q}f'(\bar{q})$:
\begin{equation}
    g(\bar{q}) = 12\bar{q}^4 + 4\bar{q}^2\ln \bar{q} + 8\bar{q}^2 + 12\bar{q}\ln \bar{q} + 2\bar{q} - 12.
\end{equation}
We evaluate the curvature of $g(\bar{q})$ by taking its second derivative:
\begin{equation*}
    g''(\bar{q}) = 144\bar{q}^2 + 8\ln \bar{q} + 28 + \frac{12}{\bar{q}}.
\end{equation*}

\vspace{0.2cm}
\noindent \textbf{Step 3: Bounding the Curvature of $g(\bar{q})$.} \\
By the standard strict convexity bound $e^x > 1+x$, setting $x = -\ln \bar{q}$ yields $\frac{1}{\bar{q}} > 1 - \ln \bar{q}$, which rearranges to the inequality $\ln \bar{q} > 1 - \frac{1}{\bar{q}}$. Applying this to the logarithmic term gives $8\ln \bar{q} > 8 - \frac{8}{\bar{q}}$. Substituting this bound into $g''(\bar{q})$ bounds the curvature as follows:
\begin{equation*}
    g''(\bar{q}) > 144\bar{q}^2 + 8 - \frac{8}{\bar{q}} + 28 + \frac{12}{\bar{q}} = 144\bar{q}^2 + \frac{4}{\bar{q}} + 36 > 0 \quad \text{for all } \bar{q} \in (0, 1).
\end{equation*}

Because $g''(\bar{q}) > 0$, the first derivative $g'(\bar{q})$ is strictly monotonically increasing. Next, we evaluate $g'(1) = 48 + 20 + 14 = 82 > 0$. As $\bar{q} \to 0^+$, the term $12\ln \bar{q}$ forces $g'(\bar{q}) \to -\infty$.

By the Intermediate Value Theorem, $g'(\bar{q})$ crosses zero exactly once at $r_1 \in (0, 1)$. Thus, $g(\bar{q})$ decreases strictly on $(0, r_1)$ and increases strictly on $(r_1, 1)$.

Next, we evaluate $g(1) = 12 + 8 + 2 - 12 = 10 > 0$. As $\bar{q} \to 0^+$, $g(\bar{q}) \to -12 < 0$. Hence, $g(\bar{q})$ crosses zero exactly once at $r_g \in (0, 1)$. 

Because $f'(\bar{q}) = g(\bar{q})/\bar{q}$, $f'(\bar{q})$ shares this exact unique root and sign profile. Consequently, $f(\bar{q})$ decreases strictly on $(0, r_g)$ and increases strictly on $(r_g, 1)$.

Now, note that $f(1) = 3 + 3 - 10 + 4 = 0$. Because $f(\bar{q})$ strictly increases to $0$ on $(r_g, 1)$, we have $f(\bar{q}) < 0$ on $(r_g, 1)$.
 As $\bar{q} \to 0^+$, the term $-12\ln \bar{q}$ forces $f(\bar{q}) \to +\infty$.

Hence, $f(\bar{q})$ crosses zero exactly once at $r_f \in (0, 1)$. Because $D^{(5)}(\bar{q}) = \frac{8}{\bar{q}^5}f(\bar{q})$, the fifth derivative $D^{(5)}(\bar{q})$ is strictly positive on $(0, r_f)$ and strictly negative on $(r_f, 1)$.

\vspace{0.2cm}
\noindent \textbf{Step 4: The Differential Cascade.} \\
We now cascade this uniqueness backward through the antiderivatives using Rolle's Theorem, taking advantage of the boundary condition $D^{(k)}(1) = 0$:
\begin{enumerate}
    \item \textbf{For $D^{(4)}(\bar{q})$:} Its derivative is $D^{(5)}(\bar{q})$. Thus, $D^{(4)}$ increases on $(0, r_f)$ and decreases on $(r_f, 1)$. Since $D^{(4)}(1) = 0$, it must be that $D^{(4)}(\bar{q}) > 0$ on $(r_f, 1)$. Because $D^{(4)}(0^+) \to -\infty$, it increases from $-\infty$ to a positive peak at $r_f$, securing exactly one root $r_4 \in (0, r_f)$.
    \item \textbf{For $D'''(\bar{q})$:} Its derivative is $D^{(4)}(\bar{q})$, which is negative on $(0, r_4)$ and positive on $(r_4, 1)$. Thus, $D'''(\bar{q})$ decreases on $(0, r_4)$ and increases on $(r_4, 1)$. Since $D'''(1) = 0$, it is strictly negative on $(r_4, 1)$. Because $D'''(0^+) \to +\infty$, it decreases from $+\infty$ to a negative trough at $r_4$, securing exactly one root $r_3 \in (0, r_4)$.
    \item \textbf{For $D''(\bar{q})$:} Its derivative is $D'''(\bar{q})$, which is first positive, then negative. Thus, $D''(\bar{q})$ increases then decreases to $D''(1) = 0$, making it strictly positive on $(r_3, 1)$. Because $D''(0^+) \to -\infty$, it increases from $-\infty$ to a positive peak, securing exactly one root $r_2 \in (0, r_3)$.
    \item \textbf{For $D'(\bar{q})$:} Its derivative is $D''(\bar{q})$, which is first negative, then positive. Thus, $D'(\bar{q})$ decreases then increases to $D'(1) = 0$, making it strictly negative on $(r_2, 1)$. Because $D'(0^+) \to +\infty$, it decreases from $+\infty$ to a negative trough, securing exactly one root $r_1' \in (0, r_2)$.
    \item \textbf{For $D(\bar{q})$:} Its derivative is $D'(\bar{q})$, which is first positive, then negative. Thus, $D(\bar{q})$ increases on $(0, r_1')$ and decreases on $(r_1', 1)$ to $D(1) = 0$. Consequently, $D(\bar{q})$ is strictly positive on $(r_1', 1)$. Because $D(0^+) \to -\infty$, it increases from $-\infty$ to a positive peak at $r_1'$, securing exactly one root $\bar{q}_c \in (0, r_1')$.
\end{enumerate}
This establishes the uniqueness of the root on $(0, 1)$.
\end{proof}

\begin{lemma}\label{lem:universal_transversality}
There exists a constant $\bar t \approx 2.13$ such that for any stage game with $E(1)<0$ the transversality constant $\mathcal{T}=\partial^2_{\rho\tilde{r}} W(1,0;\bar{x})$ at the undiscounted root $\bar{x}$ satisfies $\mathcal{T} \neq 0$ if and only if the acceleration loss $\frac{T-P}{P-S} \neq \bar t$. Furthermore:
\begin{itemize}
    \item If $\frac{T-P}{P-S} > \bar t$, then $\mathcal{T} > 0$.
    \item If $\frac{T-P}{P-S} < \bar t$, then $\mathcal{T} < 0$.
\end{itemize}
\end{lemma}

\begin{proof}
Let $\kappa = P-S > 0$. At the undiscounted root $\bar{x} \in (0,1)$, define $K_0 = \kappa + \delta \bar{x} > 0$ and the dimensionless coordinate $\bar{q} = \frac{K_0}{\kappa} = 1 + \frac{\delta \bar{x}}{\kappa}$. Because $\bar{x} \in (0, 1)$ and $\delta < 0$, we have $\bar{q} \in (0, 1)$. We substitute $K_0 = \kappa \bar{q}$ and $L = \ln(K_0/\kappa) = \ln \bar{q}$, which also implies $\delta \bar{x} = \kappa(\bar{q}-1)$.

\vspace{0.2cm}
\noindent \textbf{Step 1: Factoring the Transversality Condition.} \\
We take the closed form of $\mathcal{T}$ given in Lemma~\ref{lem:transversality-constant} and factor the following way. 
Substituting $K_0 = \kappa \bar{q}$ and $L = \ln \bar{q}$ permits factoring out $\kappa^2$ from the $(T-S)$ term and $\kappa^3$ from the remaining term. We can rewrite this compactly as:
\begin{equation} \label{eq:T_compact}
    \mathcal{T} = \frac{\kappa^2}{\delta^3} \Big[ \left(\frac{T-S}{\kappa}\right) A(\bar{q}) + B(\bar{q}) \Big],
\end{equation}
where we define the geometric polynomials:
\begin{align*}
    A(\bar{q}) &= \bar{q}^2(1-\ln\bar{q}) + 4\bar{q} - 2\ln^2\bar{q} - 5\ln\bar{q} - 5, \\
    B(\bar{q}) &= -\bar{q}^2 + 2\ln^2\bar{q} + 2\ln\bar{q} + 1.
\end{align*}

\vspace{0.2cm}
\noindent \textbf{Step 2: Incorporating the Root Condition.} \\
The root condition $E(\bar{x}) = 0$ established in Theorem A.13 is defined as:
\begin{equation*}
    (T-P) \left[ \frac{\bar{x}^2}{2\kappa} - \frac{\bar{x}}{\delta} + \frac{\kappa}{\delta^2}\ln\left(1+\frac{\delta\bar{x}}{\kappa}\right) \right] + \bar{x}^2 = 0.
\end{equation*}
Substituting $\bar{x} = \frac{\kappa}{\delta}(\bar{q}-1)$ and dividing the entire equation by $\frac{\kappa^2}{\delta^2}$ yields:
\begin{equation*}
    \frac{T-P}{\kappa} \left[ \frac{(\bar{q}-1)^2 - 2(\bar{q}-1) + 2\ln\bar{q}}{2} \right] + (\bar{q}-1)^2 = 0.
\end{equation*}
Let $C_{num}(\bar{q}) = (\bar{q}-1)^2 - 2(\bar{q}-1) + 2\ln\bar{q} = \bar{q}^2 - 4\bar{q} + 3 + 2\ln\bar{q}$. Solving for the normalized temptation spread yields the exact bijection:
\begin{equation} \label{eq:T_minus_P_f}
    \frac{T-P}{\kappa} = \frac{-2(1-\bar{q})^2}{C_{num}(\bar{q})} =: f(\bar{q}).
\end{equation}
Since $T-S = (T-P) + (P-S) = (T-P) + \kappa$, we obtain the exact substitution for $T-S$:
\begin{equation} \label{eq:T_minus_S_sub}
    \frac{T-S}{\kappa} = f(\bar{q}) + 1 = \frac{-2(1-\bar{q})^2 + C_{num}(\bar{q})}{C_{num}(\bar{q})}.
\end{equation}

\vspace{0.2cm}
\noindent \textbf{Step 3: Algebraic Expansion of the Threshold Function.} \\
Substitute \eqref{eq:T_minus_S_sub} into the compact transversality expression \eqref{eq:T_compact}:
\begin{align*}
    \mathcal{T} &= \frac{\kappa^2}{\delta^3 C_{num}(\bar{q})} \Big[ \big( C_{num}(\bar{q}) - 2(1-\bar{q})^2 \big) A(\bar{q}) + C_{num}(\bar{q}) B(\bar{q}) \Big] \nonumber \\
    &= \frac{\kappa^2}{\delta^3 C_{num}(\bar{q})} \Big[ C_{num}(\bar{q}) \big( A(\bar{q}) + B(\bar{q}) \big) - 2(1-\bar{q})^2 A(\bar{q}) \Big] =: \frac{\kappa^2}{\delta^3 C_{num}(\bar{q})} D(\bar{q}).
\end{align*}
We now expand $D(\bar{q})$ explicitly. First, we compute $A(\bar{q}) + B(\bar{q}) = -\bar{q}^2\ln\bar{q} + 4\bar{q} - 3\ln\bar{q} - 4$. Multiplying this by $C_{num}(\bar{q}) $ gives:
\begin{align*}
    C_{num}(\bar{q}) \big( A(\bar{q}) + B(\bar{q}) \big) = {}& -\bar{q}^4\ln\bar{q} + 4\bar{q}^3\ln\bar{q} + 4\bar{q}^3 - 2\bar{q}^2\ln^2\bar{q} - 6\bar{q}^2\ln\bar{q} - 20\bar{q}^2 \nonumber \\
    & + 20\bar{q}\ln\bar{q} + 28\bar{q} - 6\ln^2\bar{q} - 17\ln\bar{q} - 12.
\end{align*}
Next, we expand $2(1-\bar{q})^2 = 2\bar{q}^2 - 4\bar{q} + 2$ and multiply by $A(\bar{q})$:
\begin{align*}
    -2(1-\bar{q})^2 A(\bar{q}) = {}& -2\bar{q}^4 + 2\bar{q}^4\ln\bar{q} - 4\bar{q}^3\ln\bar{q} - 4\bar{q}^3 + 4\bar{q}^2\ln^2\bar{q} + 12\bar{q}^2\ln\bar{q} + 24\bar{q}^2 \nonumber \\
    & - 8\bar{q}\ln^2\bar{q} - 20\bar{q}\ln\bar{q} - 28\bar{q} + 4\ln^2\bar{q} + 10\ln\bar{q} + 10.
\end{align*}
Adding the two terms  cancels all cubic ($\bar{q}^3, \bar{q}^3\ln\bar{q}$) and linear ($\bar{q}, \bar{q}\ln\bar{q}$) terms:
\begin{equation*}
    D(\bar{q}) = \bar{q}^4\ln\bar{q} - 2\bar{q}^4 + 2\bar{q}^2\ln^2\bar{q} + 6\bar{q}^2\ln\bar{q} + 4\bar{q}^2 - 8\bar{q}\ln^2\bar{q} - 2\ln^2\bar{q} - 7\ln\bar{q} - 2.
\end{equation*}

\vspace{0.2cm}
\noindent \textbf{Step 4: Evaluating the Signs and the Universal Threshold.} \\
To establish the sign of $\mathcal{T}$, we evaluate its prefactor $\frac{\kappa^2}{\delta^3 C_{num}(\bar{q})}$. We have $\kappa^2 > 0$, and strategic complementarity demands $\delta < 0 \implies \delta^3 < 0$. Furthermore, taking the derivative of $C_{num}(\bar{q})$ yields $C_{num}'(\bar{q}) = 2\bar{q} - 4 + 2/\bar{q} = \frac{2(1-\bar{q})^2}{\bar{q}} > 0$ for all $\bar{q} \in (0, 1)$. Because $C_{num}(1) = 0$ and $C_{num}$ is strictly increasing, $C_{num}(\bar{q}) < 0$ on $(0, 1)$. Consequently, the prefactor is  strictly positive.

Thus, $\text{sign}(\mathcal{T}) = \text{sign}(D(\bar{q}))$. The  function $D(\bar{q})$ evaluates to $0$ at $\bar{q}=1$ and transitions exactly once from negative to positive at a unique root $\bar{q}_c \approx 0.33486$, as established in  Lemma~\ref{lem:D_unique_root}.

From \eqref{eq:T_minus_P_f}, the root condition demands that $\frac{T-P}{P-S} = f(\bar{q})$. One can verify that $f'(\bar{q}) > 0$ on $(0, 1)$, establishing $f$ as a bijection. Evaluating $f$ at the critical coordinate $\bar{q}_c$ yields the universal constant $\bar t = f(\bar{q}_c) \approx 2.130075$. Because $f$ is strictly monotonic, $\frac{T-P}{P-S} > \bar t \iff \bar{q} > \bar{q}_c \iff D(\bar{q}) > 0 \iff \mathcal{T} > 0$. The  symmetric equivalence completes the proof.
\end{proof}

Now we are ready to state the main theorem of this section.

\begin{theorem}[Global optimality via bifurcation]
\label{thm:bifurcation_global_optimality}
For $\delta<0$, define the payoff ratio $q := \frac{T-R}{P-S} \in (0, 1)$. When
\begin{equation}\label{eq:E1condi}
    \frac{T-P}{P-S} \in \left(\frac{2(1-q)^2}{4q - q^2 - 3 - 2\ln q}, \infty\right)\setminus \{\bar t\},
\end{equation}
there exists a unique root $\bar{x} \in (0,1)$ satisfying $E(\bar{x}) = 0$ (by Remark~\ref{rem:E-root-exists}), and the transversality condition $\mathcal{T}\neq 0$ holds.  When $\mathcal{T}>0$ there exists $\varepsilon > 0$ such that for every initial condition $x_0 \in ( \bar{x},\bar{x} + \varepsilon)$:
\begin{enumerate}
    \item For every discount rate $r > 0$, the meta-game admits a finite symmetric interior critical point $\gamma^*(r) = \eta^*(x_0)\, r \in (0, \infty)$.
    \item The profile $(\gamma^*(r), \gamma^*(r))$ is a strict global symmetric Nash equilibrium.
\end{enumerate}
When $\mathcal{T}<0$ the same conclusion holds for all $x_0 \in ( \bar{x} - \varepsilon,\bar{x})$.
\end{theorem}

\begin{proof}
First, we show that when the condition~\ref{eq:E1condi} is satisfied, a unique root $\bar x \in (0,1)$ for $E(x_0)$ exists. 

\begin{lemma} \label{lem:clean_E1}
 When $\delta<0$, the root condition $E(1) < 0$ holds if and only if the following inequality is satisfied:
\begin{equation*}
    \frac{T-P}{P-S} > \frac{2(1-q)^2}{4q - q^2 - 3 - 2\ln q},
\end{equation*}
where $q := \frac{T-R}{P-S} \in (0, 1)$.
\end{lemma}

\begin{proof}
Evaluating $E(1)$ yields
\begin{equation*}
    E(1) = \frac{T-P}{\kappa} \left[ \frac{1}{2} - \frac{\kappa}{\delta} + \frac{\kappa^2}{\delta^2} \ln\left(1+\frac{\delta}{\kappa}\right) \right] + 1.
\end{equation*}
 Note that $1 + \frac{\delta}{\kappa} = \frac{\kappa+\delta}{\kappa} = \frac{T-R}{P-S} = q$. Consequently, we have $\frac{\delta}{\kappa} = q - 1$. Substituting these relations into the bracketed term (which we denote by $A(q)$) gives:
\begin{align*}
    A(q) &= \frac{1}{2} - \frac{1}{q-1} + \frac{\ln q}{(q-1)^2} \nonumber \\
    &= \frac{(q-1)^2 - 2(q-1) + 2\ln q}{2(q-1)^2} \nonumber \\
    &= \frac{q^2 - 4q + 3 + 2\ln q}{2(1-q)^2}.
\end{align*}
Let the numerator be $N(q) = q^2 - 4q + 3 + 2\ln q$. We have $N(1) = 0$, and differentiating yields $N'(q) = 2q - 4 + \frac{2}{q} = \frac{2(1-q)^2}{q}$. Since $N'(q) > 0$ for all $q \in (0,1)$, it follows that $N(q) < 0$ on this interval. Thus, $A(q)$ is strictly negative for all valid parameters.

Enforcing $E(1) = \frac{T-P}{P-S} A(q) + 1 < 0$ reduces to:
\begin{equation}
    \frac{T-P}{P-S} > -\frac{1}{A(q)} = \frac{2(1-q)^2}{-q^2 + 4q - 3 - 2\ln q} = \frac{2(1-q)^2}{4q - q^2 - 3 - 2\ln q},
\end{equation}
completing the proof.
\end{proof}

Define $\Phi(\tilde{r}, x_0) := \partial_\rho W(1, \tilde{r}; x_0)$.

\textit{Reduction of the first-order condition.} Fix $\gamma_2, r > 0$ and write $\rho = \gamma_1/\gamma_2$, $\tilde{r} = r/\gamma_2$. By Lemma~\ref{lem:effective_discount}, $V_1(\gamma_1, \gamma_2; r) = P + \tilde{r}\, W(\rho, \tilde{r};x_0)$; since $\tilde{r}$ does not depend on $\gamma_1$, differentiating in $\gamma_1$ and evaluating at the symmetric profile $\gamma_1 = \gamma_2$ (so $\rho = 1$) gives
\begin{align*}
\frac{\partial V_1}{\partial \gamma_1}(\gamma_2, \gamma_2; r) = \frac{\tilde{r}}{\gamma_2}\, \partial_\rho W(1, \tilde{r};x_0) = \frac{\tilde{r}}{\gamma_2}\, \Phi(\tilde{r}, x_0).
\end{align*}
Thus a symmetric profile $(\gamma, \gamma)$ is a critical point of the discounted meta-game at rate $r$ if and only if $\Phi(r/\gamma, x_0) = 0$.

\textit{Construction of the critical point.} At $\tilde{r} = 0$, $W(\rho, 0; x_0) = J_1^0(\rho, 1; x_0)$, so $\Phi(0, x_0) = \partial_\rho J_1^0(1, 1; x_0) = E(x_0)$. By the analysis of $E$ in Lemma~\ref{lem:E-root}, $E$ has the unique zero $\bar x$ in $(0,1]$, with $E'(\bar x)<0$ and $E(x_0)<0$ for $x_0\in(\bar x,1]$; in particular $\Phi(0,\bar x)=E(\bar x)=0$.

Taking the transversality assumption $\partial^2_{\rho\tilde r}W(1,0;\bar x)\neq0$, we continue the proof assuming $\partial^2_{\rho\tilde r}W(1,0;\bar x)>0$. The argument for $\partial^2_{\rho\tilde r}W(1,0;\bar x)<0$ is analogous, and the necessary adjustment will be mentioned below. 

As $\Phi(\tilde r,x_0)=\partial_\rho W(1,\tilde r;x_0)$ is continuously differentiable on a neighbourhood of $(0,\bar x)$ by Lemma~\ref{lem:final}(ii), and $\partial_{\tilde r}\Phi(0,\bar x)=\partial^2_{\rho\tilde r}W(1,0;\bar x)>0$ by the transversality assumption, the Implicit Function Theorem yields $\varepsilon_1>0$ and a continuously differentiable function $\tilde r^*\colon(\bar x-\varepsilon_1,\bar x+\varepsilon_1)\to\mathbb{R}$ with $\Phi(\tilde r^*(x_0),x_0)=0$ and $\tilde r^*(\bar x)=0$. Its derivative at $\bar x$ is
\begin{align*}
\tilde{r}^*{}'(\bar{x}) = -\frac{\partial_{x_0}\Phi(0,\bar{x})}{\partial_{\tilde{r}}\Phi(0,\bar{x})} = -\frac{E'(\bar{x})}{\partial_{\tilde{r}}\Phi(0,\bar{x})} > 0,
\end{align*}
since $E'(\bar{x}) < 0$ and $\partial_{\tilde{r}}\Phi(0,\bar{x}) > 0$. Hence $\tilde{r}^*$ is strictly increasing through $\bar{x}$, so $\tilde{r}^*(x_0) > 0$ precisely for $x_0 \in (\bar{x}, \bar{x}+\varepsilon_1)$ (maybe reducing $\varepsilon_1$), the interval on which we work henceforth. In the case $\partial^2_{\rho\tilde r}W(1,0;\bar x)<0$, one would choose the left-sided expansion $(\bar x - \eps_1, \bar x)$ to still achieve $\tilde{r}^*{}'(\bar{x})>0$, and continue analogously.

Set $\eta^*(x_0) := 1/\tilde{r}^*(x_0)$ and $\gamma^*(r) := \eta^*(x_0)\, r$. At $\gamma_1 = \gamma_2 = \gamma^*(r)$ the effective rate is $r/\gamma^*(r) = 1/\eta^*(x_0) = \tilde{r}^*(x_0)$, so $\Phi(\tilde{r}^*(x_0), x_0) = 0$ gives $\partial_{\gamma_1} V_1(\gamma^*(r), \gamma^*(r); r) = 0$ for every $r > 0$. This proves part~(1). 

\textit{Global optimality.} Fix a bounded open neighbourhood $U'\subseteq(0,\infty)$ of $\rho=1$ and a constant $\bar r_{\mathrm c}\in(0,\bar r)$, provided by Lemma~\ref{lem:final}, such that $\partial_{\rho\rho}W$ is jointly continuous on $U'\times[0,\bar r_{\mathrm c})\times X'$. By Lemma~\ref{lem:final}(iii) and Theorem~\ref{thm:undiscounted-global-max}, $\partial_{\rho\rho}W(1,0;\bar x)=\partial_{\rho\rho}J_1^0(1,1;\bar x)<0$; by joint continuity we may shrink $U'$, $\bar r_{\mathrm c}$ and $X'$ so that $\partial_{\rho\rho}W(\cdot,\tilde r;x_0)<0$ throughout $U'$ for all $(\tilde r,x_0)\in[0,\bar r_{\mathrm c})\times X'$. Thus $W(\cdot,\tilde r;x_0)$ is strictly concave on $U'$ and has at most one critical point there. Shrinking $\varepsilon$ if necessary so that $(\bar x,\bar x+\varepsilon)\subseteq X'$ and $\tilde r^*(x_0)<\bar r_{\mathrm c}$ for $x_0\in(\bar x,\bar x+\varepsilon)$, the profile $\rho=1$ is by construction a critical point of $W(\cdot,\tilde r^*(x_0);x_0)$, hence its unique critical point in $U'$ and a strict local maximiser.

It remains to show that the global maximiser of $W(\cdot,\tilde r^*(x_0);x_0)$ over $(0,\infty)$ lies in $U'$. As $x_0\searrow\bar x$, $\tilde r^*(x_0)\to 0^+$; by Lemma~\ref{lem:uniform-convergence-W} with Remark~\ref{rem:uniform-in-x0} and by Lemma~\ref{lem:J0-x0-uniform}, $W(\cdot,\tilde r^*(x_0);x_0)\to J_1^0(\cdot,1;\bar x)$ uniformly on every compact subset of $(0,\infty)$ and on $[1,\infty)$. By Theorem~\ref{thm:undiscounted-global-max}, $J_1^0(\cdot,1;\bar x)$ attains a unique strict global maximum at $\rho=1$, is strictly increasing on $(0,1)$ and is strictly decreasing on $(1,\infty)$.

\textit{Confinement to a compact interval.} The constants in Lemma~\ref{lem:lower-tail-penalty} depend continuously on $x_0$, so may be chosen locally uniformly near $\bar x$; the lemma then provides $\underline{\rho}\in(0,\inf U')$ with $\sup_{\rho\in(0,\underline{\rho}]}W(\rho,\tilde r^*(x_0);x_0)<J_1^0(1,1;\bar x)-1<W(1,\tilde r^*(x_0);x_0)$ for $x_0$ close enough to $\bar x$, where the last inequality uses $W(1,\tilde r^*(x_0);x_0)\to J_1^0(1,1;\bar x)$.

For the upper tail, fix any $\gamma_{\mathrm{high}}>\sup U'$; since $J_1^0(\cdot, 1; \bar{x})$ is strictly decreasing on $(1, \infty)$ (Theorem~\ref{thm:undiscounted-global-max}) and $W \to J_1^0$ uniformly on $[1, \infty)$ (Lemma~\ref{lem:uniform-convergence-W}), we have $\sup_{\rho\ge\gamma_{\mathrm{high}}}W(\rho,\tilde r^*(x_0);x_0)<W(1,\tilde r^*(x_0);x_0)$ for $x_0$ close to $\bar x$. Hence the global maximiser lies in the compact interval $K := [\underline{\rho}, \gamma_{\mathrm{high}}] \supseteq\overline{U'}$.

\textit{Localisation at $\rho = 1$.} Since $K\setminus U'$ is compact and $J_1^0(\cdot,1;\bar x)$ is continuous with its unique strict maximum at $\rho=1\in U'$, the gap $\eta:=J_1^0(1,1;\bar x)-\sup_{\rho\in K\setminus U'}J_1^0(\rho,1;\bar x)$ is strictly positive. By uniform convergence on $K$, for $x_0$ close enough to $\bar x$ we have $\sup_{\rho\in K}|W(\rho,\tilde r^*(x_0);x_0)-J_1^0(\rho,1;\bar x)|<\eta/3$, whence
\begin{align*}
\sup_{\rho\in K\setminus U'}W(\rho,\tilde r^*(x_0);x_0)<J_1^0(1,1;\bar x)-\tfrac{2\eta}{3}<W(1,\tilde r^*(x_0);x_0)-\tfrac{\eta}{3}.
\end{align*}
Therefore the maximiser over $K$ lies in $U'$. Combined with the strict concavity of $W(\cdot,\tilde r^*(x_0);x_0)$ on $U'$, the maximiser over $K$ equals the unique critical point $\rho=1$, which is thus the unique strict global maximiser of $W(\cdot,\tilde r^*(x_0);x_0)$ over $(0,\infty)$.

\textit{From best response to Nash equilibrium.} Restoring the scaling equivalence of Lemma~\ref{lem:effective_discount}, $\gamma_1=\gamma^*(r)$ is the unique strict global best response of Firm~$1$ to $\gamma_2=\gamma^*(r)$ for every $r>0$. As the stage game is symmetric and $x_0=y_0$, the map $\gamma_2\mapsto V_2(\gamma^*(r),\gamma_2;r)$ is obtained from $\gamma_1\mapsto V_1(\gamma_1,\gamma^*(r);r)$ by interchanging the roles of the players, so $\gamma_2=\gamma^*(r)$ is likewise the unique strict global best response of Firm~$2$ to $\gamma_1=\gamma^*(r)$. Hence $(\gamma^*(r),\gamma^*(r))$ is a strict global symmetric Nash equilibrium. Shrinking $\varepsilon$ guarantees this for all $x_0\in(\bar x,\bar x+\varepsilon)$, proving part~(2).
\end{proof}

\subsection{Proof of Proposition~\ref{prop:delta_sensitivity}}

\begin{proof}
Fix a stage game satisfying the hypotheses, and denote its complementarity
parameter by $\delta_0<0$. Throughout the proof, hold
$\kappa$, $T-P$, $x_0=y_0$, and $r>0$ fixed. Write
$\bar x=\bar x(\delta_0)$ for the undiscounted root, define
$\Phi(\tilde r,x_0,\delta):=\partial_\rho W(1,\tilde r;x_0,\delta)$,
and let $\mathcal T_0:=\partial_{\tilde r}\Phi(0,\bar x,\delta_0)$.
By Lemma \ref{lem:universal_transversality} and the exclusion of the threshold in
Theorem \ref{thm:interior-NE-informal}, $\mathcal T_0\ne0$.

We first justify joint regularity, including dependence on $\delta$.
The  representation \eqref{eq:dW-minus} in the proof of Lemma~\ref{lem:final} gives
\[
  \Phi(\tilde r,x_0,\delta)
  =\int_0^{\tau(x_0,\delta)}
     e^{-\tilde r t}G(t,x_0,\delta)\,dt,
\]
where $\tau(x_0,\delta)=\frac{2}{\delta}
\ln(1+\delta x_0/\kappa)$ and
$G(t,x_0,\delta)=\partial_\rho F_1(t,1;x_0,\delta)$ is computed from
the unprojected interior trajectories. Near $(\bar x,\delta_0)$,
we have $\delta<0$ and $\kappa+\delta x_0>0$; hence $\tau$ and $G$
are smooth, and $\tau$ is locally bounded. The substitution
$t=s\tau(x_0,\delta)$ fixes the integration interval at $[0,1]$,
so differentiation under the integral shows that $\Phi$ is jointly
continuously differentiable in $(\tilde r,x_0,\delta)$. The same
representation extends $\Phi$ to an open interval of $\tilde r$ about
zero. In particular, $\partial_{\tilde r}\Phi$ and
$\partial_\delta\Phi$ are jointly continuous near $(0,\bar x,\delta_0)$.

At zero effective discounting, $\Phi(0,x_0,\delta)=E(x_0;\delta)$ given in \eqref{eq:E-def}.
Set $z=\delta x_0/\kappa$ and $g(z)=[\ln(1+z)-z]/z^2$.
The inequalities $\delta<0$ and
$\kappa+\delta x_0\geq\kappa+\delta=T-R>0$ imply $z\in(-1,0)$.
The expression for $E$ can be written as
$$E(x_0;\delta)=\frac{T-P}{\kappa}x_0^2\left[\frac12+g(z)\right]+x_0^2.
$$
To determine its derivative, write $g'(z)=N(z)/z^3$, where
$N(z)=2z-z^2/(1+z)-2\ln(1+z)$.
Since $N'(z)=z^2/(1+z)^2>0$ and $N(0)=0$, we have $N(z)<0$ on
$(-1,0)$ and therefore $g'(z)>0$. It follows that
$\partial_\delta E(x_0;\delta)
=\frac{(T-P)x_0^3}{\kappa^2}g'(z)>0$.
In particular,
$\partial_\delta\Phi(0,\bar x,\delta_0)
=\partial_\delta E(\bar x;\delta_0)>0$, where the initial condition
is held fixed in this partial derivative.

Joint continuity and $\mathcal T_0\ne0$ therefore give a neighbourhood
$\mathcal U$ of $(0,\bar x,\delta_0)$ on which
\[
  \left|\partial_{\tilde r}\Phi(\tilde r,x_0,\delta)
             -\mathcal T_0\right|
  <\frac{|\mathcal T_0|}{2},
  \qquad
  \partial_\delta\Phi(\tilde r,x_0,\delta)>0.
\]
The equilibrium branch constructed in Appendix \ref{sec:discounted} 
satisfies
$\Phi(\tilde r^*(x_0;\delta_0),x_0,\delta_0)=0$,
$\tilde r^*(x_0;\delta_0)>0$, and
$\tilde r^*(x_0;\delta_0)\to0$ as $x_0\to\bar x$ along its
admissible side. After shrinking the nonempty open equilibrium
neighbourhood $\mathcal N$ if necessary, its equilibrium points
$(\tilde r^*(x_0;\delta_0),x_0,\delta_0)$ therefore lie in
$\mathcal U$ for every $x_0\in\mathcal N$.
At each such point, $\partial_{\tilde r}\Phi$ has absolute value
greater than $|\mathcal T_0|/2$ and has the same sign as
$\mathcal T_0$.

Now fix $x_0\in\mathcal N$, set
$\tilde r_0:=\tilde r^*(x_0;\delta_0)$, and write
$P_*:=(\tilde r_0,x_0,\delta_0)$ for the equilibrium evaluation point.
The implicit function theorem at $P_*$ provides a continuously
differentiable local branch $\delta\mapsto\tilde r^*(x_0;\delta)$.
After reducing the parameter interval, this branch remains strictly
positive and its graph remains in $\mathcal U$. The strict concavity
and uniform payoff-gap estimates used to establish global optimality
in Appendix \ref{sec:discounted} persist under sufficiently small changes in $\delta$:
the trajectories, hitting times, and constants in the tail bounds
depend continuously on the game parameters. Thus the branch continues
to represent the symmetric equilibrium.

Differentiating the equilibrium condition with $x_0$ fixed and using
$\gamma^*(r,x_0;\delta)=r/\tilde r^*(x_0;\delta)$ yields
\[
\begin{aligned}
  \left.\frac{\partial\tilde r^*(x_0;\delta)}{\partial\delta}
       \right|_{\delta=\delta_0}
  &=-\frac{\partial_\delta\Phi(P_*)}
           {\partial_{\tilde r}\Phi(P_*)},\\[4pt]
  \left.\frac{\partial\gamma^*(r,x_0;\delta)}{\partial\delta}
       \right|_{\delta=\delta_0}
  &=\frac{r}{\tilde r_0^2}
     \frac{\partial_\delta\Phi(P_*)}
          {\partial_{\tilde r}\Phi(P_*)}.
\end{aligned}
\]
The numerator is positive, while the denominator has the sign of
$\mathcal T_0$. Hence the derivative of the equilibrium learning rate
has the sign of $\mathcal T_0$. Lemma \ref{lem:universal_transversality}  gives $\mathcal T_0>0$ when
$(T-P)/(P-S)>\bar t$ and $\mathcal T_0<0$ when
$(T-P)/(P-S)<\bar t$, proving the two claims.
The base parameter $\delta_0$ was arbitrary among games satisfying
the hypotheses.
\end{proof}

\end{document}